\documentclass{article}
\pdfoutput=1
\usepackage{kr}

\usepackage{times}
\usepackage{soul}
\usepackage{url}
\usepackage[hidelinks]{hyperref}
\usepackage[utf8]{inputenc}
\usepackage[small]{caption}
\usepackage{graphicx}
\usepackage{amsmath}
\usepackage{amsthm}
\usepackage{booktabs}
\usepackage{algorithm}
\newtheorem{example}{Example}
\newtheorem{theorem}{Theorem}

\title{Forward \LTLf Synthesis: DPLL At Work}

\author{%
    Marco Favorito\footnote{The author's views are his own, and they do not reflect those of his employer.}
    \affiliations
    Bank of Italy
    \emails
    marco.favorito@gmail.it    \\
    marco.favorito@bancaditalia.it    
}

\usepackage{xspace}
\usepackage{mathtools}
\usepackage{amssymb}
\usepackage[noend]{algpseudocode}
\usepackage{algorithmicx}
\usepackage{cleveref}
\usepackage{subcaption}
\usepackage{xcolor}
\usepackage{paralist}
\usepackage{multicol}
\usepackage{amsmath}
\usepackage{amsthm}
\usepackage{enumerate}
\usepackage{todonotes}

\hypersetup{pdfauthor={Marco Favorito},
            pdftitle={Forward LTLf Synthesis: DPLL At Work},
            pdfsubject={LTLf Synthesis},
            pdfkeywords={Linear Temporal Logic on Finite Traces, LTLf Synthesis, AND-OR Graph Search}}

\newcommand{\myi}{(\emph{i})\xspace}
\newcommand{\myii}{(\emph{ii})\xspace}
\newcommand{\myiii}{(\emph{iii})\xspace}

 \newcommand{\B}{\mathcal{B}}

\newcommand{\I}{\mathcal{I}}

\renewcommand{\O}{\mathcal{O}} \renewcommand{\P}{\mathcal{P}}
 \newcommand{\R}{\mathcal{R}}
\renewcommand{\S}{\mathcal{S}}

\newcommand{\V}{\mathcal{V}}
\newcommand{\U}{\mathcal{U}}
 \newcommand{\X}{\mathcal{X}}
\newcommand{\Y}{\mathcal{Y}} \newcommand{\Z}{\mathcal{Z}}

\newcommand{\Wnext}{\raisebox{-0.27ex}{\LARGE$\bullet$}}
\newcommand{\Next}{\raisebox{-0.27ex}{\LARGE$\circ$}}
\newcommand{\lUntil}{\mathop{\U}}
\newcommand{\Release}{\mathop{\R}}

\newcommand{\true}{\mathit{true}}

\newcommand{\false}{\mathit{false}}
\newcommand{\ttrue}{\mathtt{true}}
\newcommand{\ffalse}{\mathtt{false}}

\newcommand{\LTL}{{\sc ltl}\xspace}
\newcommand{\LTLf}{{\sc ltl}$_f$\xspace}
\newcommand{\ltlf}{{\sc ltl}$_f$\xspace}

\newcommand{\DFA}{{\sc dfa}\xspace}

\newcommand{\NNF}{{\sc nnf}\xspace}

\newcommand{\Nat}{{\rm I\kern-.23em N}}

\renewcommand{\ttrue}{\mathit{tt}}
\renewcommand{\ffalse}{\mathit{ff}}

\newcommand{\citealt}[1]{\citeauthor{#1} \citeyear{#1}}

\newcommand{\PA}{{\mathsf{pa}}\xspace}
\newcommand{\pa}{{\mathsf{pa}}\xspace}
\newcommand{\cl}{{\mathsf{cl}}\xspace}
\newcommand{\xnf}{{\mathsf{xnf}}\xspace}
\newcommand{\XNF}{{\sc xnf}\xspace}
\newcommand{\stripnext}{\textsc{RmNext}\xspace}
\newcommand{\fp}{{\mathsf{fp}}\xspace}

\newcommand{\expand}{\textsc{Expand}\xspace}
\newcommand{\sddprime}{\ensuremath{\mathsf{prime}\xspace}}
\newcommand{\sddsub}{\ensuremath{\mathsf{sub}\xspace}}

\newcommand*\Let[2]{\State #1 $\gets$ #2}

\newcommand{\pt}{{\mathsf{path}}\xspace}

\newcommand{\andnode}{\textit{AndNd}}

\newcommand{\lisa}{{$\mathsf{Lisa}$}\xspace}
\newcommand{\lydia}{{$\mathsf{Lydia}$}\xspace}
\newcommand{\cynthia}{{$\mathsf{Cynthia}$}\xspace}
\newcommand{\ltlfsyn}{{$\mathsf{Ltlfsyn}$}\xspace}
\newcommand{\nike}{{$\mathsf{Nike}$}\xspace}
\newcommand{\nikeparallel}{{$\mathsf{Nike}\text{-}\mathsf{P}$}\xspace}
\newcommand{\syfco}{{$\mathsf{Syfco}$}\xspace}

\newtheorem{lemma}{Lemma}
\newtheorem{proposition}{Proposition}

\newcommand{\EquivalenceCheck}{\textsc{EquivalenceCheck}\xspace}
\newcommand{\GetArcs}{\textsc{GetArcs}\xspace}
\newcommand{\InPath}{\textsc{InPath}\xspace}
\newcommand{\FormulaOfNode}{\textsc{FormulaOfNode}\xspace}
\newcommand{\BddRepresentation}{\textsc{BddRepresentation}\xspace}
\newcommand{\BddBasedEqCheck}{\textsc{BddBasedEqCheck}\xspace}
\newcommand{\DpllGetArcs}{\textsc{DpllGetArcs}\xspace}
\newcommand{\DpllGetOrArcs}{\textsc{DpllGetOrArcs}\xspace}
\newcommand{\DpllGetAndArcs}{\textsc{DpllGetAndArcs}\xspace}
\newcommand{\GetBranchingLiteral}{\textsc{GetBranchingLiteral}\xspace}
\newcommand{\Replace}{\textsc{Replace}\xspace}
\newcommand{\yield}{\ensuremath{\mathbf{yield}}\xspace}
\newcommand{\yieldfrom}{\ensuremath{\mathbf{yield\ from}}\xspace}

\newcommand{\movetype}{\ensuremath{\mathsf{move}}\xspace}
\newcommand{\nodetype}{\ensuremath{\mathsf{node}}\xspace}

\newcommand{\truefirst}{\ensuremath{\mathsf{TF}}\xspace}
\newcommand{\falsefirst}{\ensuremath{\mathsf{FF}}\xspace}
\newcommand{\randomstrategy}{\ensuremath{\mathsf{Rand}}\xspace}
\newcommand{\bddbasedeqcheck}{\ensuremath{\mathsf{BDD}}\xspace}
\newcommand{\hashconsingbasedeqcheck}{\ensuremath{\mathsf{Hash}}\xspace}

\newcommand{\HashConsingBasedEqCheck}{\textsc{HashConsingEqCheck}\xspace}

\newcommand{\tf}{\ensuremath{\mathsf{tf}}\xspace}
\algrenewcommand\algorithmicindent{0.4cm}%

\newcommand{\Ftt}{\ensuremath{\Diamond true}\xspace}
\newcommand{\Gff}{\ensuremath{\Box false}\xspace}

\usepackage{tikz}
\usetikzlibrary{automata, positioning, arrows}
\tikzset{->,
    >=stealth,
    node distance=3cm,
    every state/.style={thick}
}

\begin{document}

\maketitle

\begin{abstract}
This paper proposes a new AND-OR graph search framework for synthesis of Linear Temporal Logic on finite traces (\LTLf), that overcomes some limitations of previous approaches. Within such framework, we devise a procedure inspired by the Davis-Putnam-Logemann-Loveland (DPLL) algorithm to generate the next available agent-environment moves in a truly depth-first fashion, possibly avoiding exhaustive enumeration or costly compilations. We also propose a novel equivalence check for search nodes based on syntactic equivalence of state formulas. Since the resulting procedure is not guaranteed to terminate, we identify a stopping condition to abort execution and restart the search with state-equivalence checking based on Binary Decision Diagrams (BDD), which we show to be correct. The experimental results show that in many cases the proposed techniques outperform other state-of-the-art approaches.
Our implementation \nike competed in the \LTLf Realizability Track in the 2023 edition of SYNTCOMP, and won the competition.
%
%
\end{abstract}

\section{Introduction}
\emph{Program synthesis} is the task of finding a program that provably satisfies a given high-level formal specification~\cite{church1963application}.
A commonly used logic for program synthesis is Linear Temporal
Logic (\LTL) \cite{Pnu77,PnueliR89}, typically used also in model checking~\cite{BaierKatoenBook08}.
\emph{\LTL on finite traces} (\LTLf) \cite{DegVa13}, a variant of \LTL to specify \emph{finite}-horizon temporal properties, has been recently proposed as specification language for temporal synthesis~\cite{DegVa15}. 
The \LTLf synthesis setting considers a set of variables controllable by the agent, a (disjoint) set of variables controlled by the environment, and a \LTLf specification that specifies which finite traces over such variables are desirable.
The problem of \LTLf synthesis consists in finding a finite-state controller 
that, at every time step, given the values of the environment variables in the history so far, sets the next values for each agent proposition such that the generated traces comply with the \LTLf specification.

The basic technique for solving \LTLf synthesis amounts to constructing a deterministic finite automaton (\DFA) corresponding to the \LTLf specification,
and then considering it as a game arena where the agent tries
to get to an accepting state regardless of the environment's moves. 
A \emph{winning strategy}, i.e. a finite controller returned by the 
procedure, can be obtained through a backward fixpoint computation for \emph{adversarial reachability} of the \DFA accepting state.

\noindent
\textbf{Related works.} 
State-of-the-art tools such as \lydia~\cite{GiacomoF21} and \lisa~\cite{lisa} are based on the classical approach.
The main drawback of this technique is that it requires to compute the entire \DFA of the \LTLf specification, which in the worst case can be doubly exponential in the size of the formula. Therefore, the \DFA construction step becomes the main bottleneck.

A natural idea is to consider a forward search approach that expands the arena on-the-fly while searching for a solution, possibly avoiding the construction of the entire arena.
Forward-based approaches are at the core of the best solution methods designed for other AI problems:  
Planning with fully observable non-deterministic domains (FOND)~\cite{ghallab2004automatedplanning,GeBo13,CimattiRT98,Cimatti03}, where the agent has to reach the goal, despite that the environment may choose adversarially the effects of the agent actions, and Planning in partially observable nondeterministic domains (POND), also known as \emph{contingent planning}, where the search procedure must be performed over the \emph{belief-states}~\cite{Reif84,GoBo96,BertoliCRT06}. However, techniques developed for such problems cannot be applied to ours: in a FOND planning problem, represented with PDDL~\cite{2019Haslum}, the search space is at most single-exponential~\cite{Rintanen04a}, whereas for \LTLf synthesis the state space can be of double-exponential size wrt the size of the formula;
hence, we do not rely on an encoding into PDDL, as~\cite{CamachoBMM18,CamachoM19}, which may result in a PDDL specification with exponential size. 
In 
POND planning,
despite the double-exponential size of the state space,
belief-states have a specific structure~\cite{BertoliCRT06,ToPS09}, 
so their solution techniques
cannot be directly applied to \LTLf synthesis.

For these reasons, researchers have been looking into forward search techniques specifically conceived for solving \LTLf synthesis.
Two notable attempts in this direction have been presented in \cite{XiaoLZSPV21} and \cite{GiacomoFLVX022}. The former work presents an on-the-fly synthesis approach via conducting a so-called Transition-based Deterministic Finite Automata (TDFA) game, where the acceptance condition is defined on transitions, instead of states. The main issue of that approach is the full enumeration of agent-environment moves, which are exponentially many in the number of variables. Moreover, due to the fact that the acceptance condition is defined on transitions, every generated transition has to be checked for acceptance. 
%
The latter work instead proposes a search framework for \LTLf synthesis, where the \DFA arena is seen as an AND-OR graph, and the available moves are found according to the formula associated with the current search node, by means of a Knowledge Compilation (KC) technique: Sentential Decision Diagrams (SDD)~\cite{Darwiche11}. 
Notably, they are able to branch on propositional formulas, representing several evaluations, instead of individual ones. This can drastically reduce the branching factor.
Nevertheless, for certain types of problem instances, the approach can get stuck with demanding compilations of the state formulas, needed \emph{both} for state equivalence checking and for search node expansion. Moreover, the requirement of having irreducible representation of agent-environment moves can be of little usefulness if the branching factor of the search problem is already high, hence resulting in an even greater compilation overhead. 

We think there is the need of a search approach that scales well with the increase of computational power, and that uses such power for actually exploring the search space, rather than wasting time either slavishly enumerating the exponentially many variable assignments, or by finding the minimal representation of the available search moves.

\noindent
\textbf{Contributions.} 
Our contributions are the following. First, we identify limitations of the previous AND-OR graph search framework, based on the \expand function, and propose a more general and versatile search framework for \LTLf synthesis, based on two primitive operations: \emph{state-equivalence checking} and \emph{search node expansion}. 
Then, we formalize and discuss two well-known instances of equivalence checks; one based on knowledge compilation, and the other on a computationally-cheap syntactical equivalence between state formulas.
Furthermore, we propose a novel search graph expansion technique, based on a procedure inspired by the famous Davis-Putnam-Logemann-Loveland (DPLL) algorithm.
Finally, we describe our implementation in a new tool called \nike, and compare its performance on known benchmarks with other state-of-the-art tools, showing its surprising effectiveness.
\nike won the \LTLf Realizability Track in the 2023 edition of SYNTCOMP \footnote{http://www.syntcomp.org/syntcomp-2023-results/}.

\section{Preliminaries}\label{sec:pre}
\noindent
\textbf{\LTLf Basics.}
Linear Temporal Logic over finite traces, called \LTLf~\cite{DegVa13} is a variant of Linear Temporal Logic~(\LTL)~\cite{BaierKatoenBook08} that is interpreted over finite traces rather than infinite traces~(as in \LTL). Given a set of propositions $\P$, the syntax of \ltlf is identical to \LTL, and defined as (wlog, we require that \LTLf formulas are in Negation Normal Form~(\NNF), i.e., negations only occur in front of atomic propositions): 
$
\varphi ::= \ttrue \mid \ffalse \mid p \mid \neg p \mid \varphi_1 \wedge \varphi_2 \mid \varphi_1 \vee \varphi_2 \mid  \Next \varphi \mid \Wnext \varphi  \mid \varphi_1 \lUntil \varphi_2 \mid \varphi_1 \Release \varphi_2.
$
$\ttrue$ is always true, $\ffalse$ is always false; $p \in \P$ is an \textit{atom}, and $\lnot p$ is a \textit{negated atom} (a literal $l$ is an atom or the negation of an atom); $\land$~(And) and $\lor$~(Or) are the Boolean connectives; and 
$\Next$~(Next), $\Wnext$~(Weak Next), $\lUntil$~(Until) and $\Release$~(Release) are temporal connectives. We use the usual abbreviations $\true \equiv p \lor \lnot p$, $\false \equiv p\land \lnot p$, $\Diamond \varphi \equiv \true \lUntil \varphi$ and $\Box \varphi \equiv \false \Release \varphi$. 
Also for convenience we consider traces $\rho \in (2^\P)^*$, i.e., we consider also empty traces $\epsilon$ as in~\cite{BrafmanGP18}. More specifically, a trace $\rho = \rho[0], \rho[1], \ldots \in (2^\P)^*$ is a finite sequence, where $\rho[i]~(0 \leq i < |\rho|)$ denotes the $i$-th interpretation of $\rho$, which can be considered as the set of propositions that are $true$ at instant $i$, and $|\rho|$ represents the length of $\rho$. 
We have that
$\epsilon \models \varphi$ if $\varphi$ is $\ttrue$, an $\Release$-formula or $\Wnext$-formula, hence $\epsilon \models \Box false$. $\epsilon \not\models\varphi$ if $\varphi$ is $\ffalse$, a literal, $\lUntil$-formula or $\Next$-formula, hence $\epsilon \not\models \Diamond \true$.

We consider the semantics of \LTLf as presented in~\cite{BrafmanGP18}, that also works for empty traces.
Given a finite trace $\rho$ and an \ltlf formula 
$\varphi$, we inductively define when $\varphi$ is $true$ for $\rho$ at point $i$ ($0 \leq i < |\rho|$), written $\rho, i \models \varphi$, as follows: 

\noindent- $\rho, i \models \ttrue$ and $\rho, i \not\models \ffalse$;

\noindent- $\rho, i \models p$ iff $p \in \rho[i]$, and $\rho, i \models \neg p$ iff $p \notin \rho[i]$;
	
\noindent- $\rho, i \models \varphi_1 \wedge \varphi_2$ iff $\rho,i \models \varphi_1$ and $\rho, i \models \varphi_2$;

\noindent- $\rho, i \models \varphi_1 \vee \varphi_2$ iff $\rho,i \models \varphi_1$ or $\rho, i \models \varphi_2$;
	
\noindent- $\rho, i \models \Next \varphi$ iff $0 \leq i < |\rho|-1$ and $\rho, i+1 \models \varphi$;
    
\noindent- $\rho, i \models \Wnext \varphi$ iff $0 \le i < |\rho|$ implies $\rho, i+1 \models \varphi$;

\noindent- $\rho, i \models \varphi_1 \lUntil \varphi_2$ iff there exists $j$ with $i\leq j < |\rho|$ such that $\rho,j\models \varphi_2$, and $\forall k.i \leq k < j$ we have $\rho, k \models \varphi_1$;
 
\noindent- $\rho, i \models \varphi_1 \Release \varphi_2$ iff for all $j$ with $i\leq j < |\rho|$ either we have $\rho,j\models \varphi_2$, or $\exists k.i \leq k < j$ we have $\rho, k \models \varphi_1$.

\noindent
An \LTLf formula $\varphi$ is $true$ for $\rho$, denoted by $\rho \models \varphi$, if and only if $\rho, 0 \models \varphi$. In particular, $\epsilon\models\Box false$ and $\epsilon\not\models\Diamond true$. 

We denote by $\cl(\varphi)$ the set of subformulas of $\varphi$, including $\ttrue$ and $\ffalse$. 
We denote by $\PA(\varphi) \subseteq \cl(\varphi)$ the set of literals and temporal subformulas of $\varphi$ whose primary connective is temporal~\cite{LiRPZV19}. Formally, for an \LTLf formula $\varphi$ in \NNF, we have $\PA(\varphi) = \{\varphi\}$ if $\varphi$ is a literal or temporal formula; and $\PA(\varphi) = \PA(\varphi_1) \cup \PA(\varphi_2)$ if $\varphi = (\varphi_1 \wedge \varphi_2)$ or $\varphi = (\varphi_1 \vee \varphi_2)$. 
Having \LTLf formula $\varphi$, replacing every temporal formula $\psi \in \PA(\varphi)$ with a propositional variable $a_\psi$ gives us a propositional formula $\varphi^p$; we call this operation \emph{propositionalization of $\varphi$}.
Note that $\varphi^p\in \B^+(\cl(\varphi))$, i.e. $\varphi^p$ is a positive Boolean formula over variables $\cl(\varphi)$.
Let $\phi=\varphi^p$, we denote with $\phi^{\tf}=\varphi$ the inverse operation of $\cdot^p$. Two formulas $\varphi_1$ and $\varphi_2$ are propositionally equivalent, denoted by $\varphi_1 \sim_p \varphi_2$, if,  $C \models \varphi_1^p \leftrightarrow C \models \varphi_2^p$ holds for every propositional assignment $C \in 2^{\pa(\varphi_1) \cup \pa(\varphi_2)}$. 

An \LTLf formula $\varphi$ is in neXt Normal Form~(\XNF) if $\PA(\varphi)$ only includes literals, \Next- and \Wnext-formulas. For an \LTLf formula $\varphi$ in \NNF, we can obtain its \XNF by transformation function $\xnf(\varphi)$, defined as follows:

\noindent-$~\xnf(\varphi)\!=\!\varphi$ if $\varphi$ is a literal, $\Box false$, $\Diamond true$,  \Next-, \Wnext-formula; 

\noindent-$~\xnf(\varphi_1 \wedge \varphi_2) = \xnf(\varphi_1) \wedge \xnf(\varphi_2)$; 

\noindent-$~\xnf(\varphi_1 \vee \varphi_2) = \xnf(\varphi_1) \vee \xnf(\varphi_2)$; 

\noindent-$~\xnf(\varphi_1\!\lUntil\!\varphi_2)\!=\!(\xnf(\varphi_2) \wedge \Diamond \true) \vee (\xnf(\varphi_1) \wedge \Next(\varphi_1\!\lUntil\!\varphi_2))$; 

\noindent-$\xnf(\varphi_1\!\Release\!\varphi_2)\!=\!(\xnf(\varphi_2) \vee \Box \false) \wedge (\xnf(\varphi_1) \vee \Wnext(\varphi_1\!\Release\!\varphi_2))$

\noindent
Note that $\Diamond \true$~(resp.~$\Box \false$) guarantees that empty trace is not~(resp.~is) accepted by $\lUntil$-formulas~(resp.~$\Release$-formulas).
\begin{theorem}[\citealt{LiRPZV19}]\label{thm:xnf}
Every \LTLf formula $\varphi$ in \NNF can be converted, with linear time in the formula size, to an equivalent formula in \XNF, denoted by $\xnf(\varphi)$.
\end{theorem}

\noindent
\textbf{\LTLf Formula Progression~\cite{GiacomoFLVX022}.}
Consider an \LTLf formula $\varphi$ over $\P$ and a finite trace $\rho = \rho[0], \rho[1], \ldots \in (2^\P)^*$, in order to have $\rho \models \varphi$, we can start from $\varphi$, progress or push $\varphi$ through $\rho$. The idea behind \emph{formula progression} is to split an \LTLf formula $\varphi$ into a requirement about \emph{now} $\rho[i]$, which can be checked straightaway, and a requirement about the future that has to hold in the yet unavailable suffix. That is to say, formula progression looks at $\rho[i]$ and $\varphi$, and progresses a new formula $\fp(\varphi, \rho[i])$ such that $\rho, i \models \varphi$ iff $\rho, i+1 \models \fp(\varphi, \rho[i])$. This procedure is analogous to DFA reading trace $\rho$, where reaching accepting states is essentially achieved by taking one transition after another. Formula progression has been studied in prior work, cf.~\cite{Emerson1990TemporalAM,BacchusK98}. 
Here we use the formalization provided in ~\cite{GiacomoFLVX022}.

Note that, since $\rho$ is a finite trace, it is necessary to clarify when the trace ends. To do so, two new formulas are introduced: $\Box\false$ and $\Diamond\true$, which, intuitively, refer to \emph{finite trace ends} and \emph{finite trace not ends}, respectively. For simplicity, we enrich $\cl(\varphi)$, the set of proper subformulas of $\varphi$, to include them such that $\cl(\varphi)$ is reloaded as $\cl(\varphi) \cup \cl(\Diamond\true) \cup \cl(\Box\false)$. 

\noindent
For an \LTLf formula $\varphi$ in \NNF, the \emph{progression function} $\fp(\varphi, \sigma)$, where $\sigma \in 2^\P$, is defined as follows:

\noindent- $\fp(\ttrue, \sigma) = \ttrue$ and $\fp(\ffalse, \sigma) = \ffalse$;

\noindent- $\fp(p, \sigma) = \ttrue$ if $p \in \sigma$, otherwise $\ffalse$;

\noindent- $\fp(\Diamond(true),\sigma) = \ttrue$ and $\fp(\Box(false),\sigma) = \ffalse$;

\noindent- $\fp(\neg p, \sigma) = \ttrue$ if $p \notin \sigma$, otherwise $\ffalse$;

\noindent- $\fp(\varphi_1 \wedge \varphi_2, \sigma) = \fp(\varphi_1, \sigma) \wedge \fp(\varphi_2, \sigma)$;

\noindent- $\fp(\varphi_1 \vee \varphi_2, \sigma) = \fp(\varphi_1, \sigma) \vee \fp(\varphi_2, \sigma)$;

\noindent- $\fp(\Next \varphi, \sigma) = \varphi \wedge \Diamond\true$, and $\fp(\Wnext \varphi, \sigma) = \varphi \vee \Box\false$;

\noindent-$\fp(\varphi_1\!\lUntil\!\varphi_2,\sigma)\!=\!\fp(\varphi_2, \sigma)\vee(\fp(\varphi_1,\sigma)\wedge\fp(\Next(\varphi_1\!\lUntil\!\varphi_2), \sigma))$;

\noindent-$\fp(\varphi_1\!\Release\!\varphi_2,\sigma)\!=\!\fp(\varphi_2, \sigma)\wedge(\fp(\varphi_1,\sigma)\vee\fp(\Wnext(\varphi_1\!\Release\!\varphi_2),\sigma))$

\noindent
Note that $\fp(\varphi, \sigma)$ is a positive Boolean formula on $\cl(\varphi)$, i.e.,
$\fp(\varphi, \sigma)\!\in \!\B^+(\cl(\varphi))$.
The following two propositions show that $\fp(\varphi, \sigma)$ strictly follows \LTLf semantics and retains the propositional behavior of \LTLf formulas.
\begin{proposition}[\citealt{GiacomoFLVX022}]\label{prop:fprog}
    Let $\varphi$ be an \LTLf formula over $\P$ in \NNF, $\rho$ be a finite nonempty trace, $\fp(\varphi, \sigma)$ be as above. 
    Then
    $\rho, i \models \varphi$ iff $\rho, i+1 \models \fp(\varphi, \rho[i])$.
\end{proposition}
\begin{proposition}[\citealt{GiacomoFLVX022}]\label{prop:propequiv}
Let $\varphi$ and $\psi$ be two \LTLf formulas over $\P$ in \NNF s.t. $\varphi\!\sim_p\!\psi$, and $\sigma\!\in\!2^\P$. Then $\fp(\varphi, \sigma)\!\!\sim_p\!\!\fp(\psi, \sigma)$ holds.
\end{proposition}

We generalize \LTLf formula progression from single instants to finite traces by defining $\fp(\varphi,\epsilon)=\varphi$, and 
$\fp(\varphi, \sigma u) =
\fp(\fp(\varphi, \sigma), u)$, where $\sigma \in 2^\P$ and $u \in (2^\P)^*$.
\begin{proposition}[\citealt{GiacomoFLVX022}]\label{prop:acc}
    Let $\varphi$ be an \LTLf formula over $\P$ in \NNF, $\rho$ be a finite trace. We have that $\rho \models \varphi$ iff $\epsilon \models \fp(\varphi, \rho)$.
\end{proposition}

We take the definition of the \emph{remove-next} function $\stripnext$ from \cite{GiacomoFLVX022}, defined over propositionalized \LTLf formulas in \XNF, $\varphi^p$:

\noindent- $\stripnext(\Diamond true) = \ttrue$, $\stripnext(\Box false) = \ffalse$

\noindent- $\stripnext(\varphi_1 \wedge \varphi_2) = \stripnext(\varphi_1) \wedge \stripnext(\varphi_2)$

\noindent- $\stripnext(\varphi_1 \vee \varphi_2) = \stripnext(\varphi_1) \vee \stripnext(\varphi_2)$

\noindent- $\stripnext(\Next\varphi)=\varphi\wedge\Diamond\mathit{true}$, $\stripnext(\Wnext\varphi)=\varphi\vee\Box\mathit{false}$

Note that $\stripnext$ applies to neither $\lUntil$-,$\Release$- formulas, since they do not appear in \XNF, nor literals~($p$, $\lnot p$), as the input of the function is a propositionalized \LTLf formula in \XNF form.
We have the following proposition:

\begin{proposition}[\citealt{GiacomoFLVX022}]\label{prop:xnf}
Given an \LTLf formula $\varphi$ in \NNF, $\forall \sigma\in 2^\P, \fp(\varphi, \sigma) \equiv \stripnext(\xnf(\varphi)^p|_\sigma)$,
where $\xnf(\varphi)^p|^\P_\sigma$ stands for substituting in $\xnf(\varphi)^p$ the variable $p$ with $\top$ if $p\in\sigma$ and $\bot$ if $p\in\P\setminus\sigma$.
\end{proposition}

\noindent
\textbf{AND-OR Graph Search.}                                                
Being a popular topic in AI, AND-OR graph search has attracted extensive studies. Following~\cite{Nilsson71,Nilsson82}, an AND/OR graph can be considered as a generalization of a directed graph, where there are a set of nodes $\V$ and generalized connectors~(edges) between nodes. Every connector links one single node $v \in \V$ to a set of nodes $V \subseteq \V$.
A connector is called an AND~(resp. OR) connector, if there is a logical AND~(resp. OR) relationship among $V$.
It should be noted that in this work we only focus on specific AND-OR graphs, where every node has only one connector leading to its successor nodes. Therefore, we have AND-nodes with an AND connector, and OR-nodes with an OR connector. Moreover, the set of goal nodes $V_g$ only consists of OR-nodes.

The AND-OR graph search problem was first introduced in~\cite{Nilsson71}. Intuitively speaking, the searching procedure aims to find a winning plan that encodes a path leading from the initial node to goal nodes. It is possible to involve both kinds of nodes in the winning plan, therefore, the plan lists one outgoing option at OR-nodes, and all outgoing options at AND-nodes leading to branches. Hence, a winning plan is essentially a tree such that all leaves are goal nodes. There has been extensive studies on AND-OR graph search techniques~\cite{MahantiB85,Chakrabarti94,JimenezT00}, and have been utilized in a lot of applications, e.g., FOND planning~\cite{MyND,Mattmuller13,GeBo13}.

\noindent
\textbf{Knowledge Compilation: BDDs and SDDs.}
Knowledge Compilation~\cite{DarwicheM02}
is a family of approaches for dealing with the
computational intractability of general propositional reasoning.
A propositional theory is compiled off-line into a target language, which is then used on-line to answer a large number of queries in polynomial time. The key motivation behind knowledge compilation is to push as much of the computational overhead into
the off-line phase, which is amortized over all on-line queries.
There are a plethora of knowledge compilation techniques.
Perhaps the first knowledge compilation technique are (Ordered) Binary Decision Diagrams (BDDs) \cite{Bryant92}, where
in order to represent a Boolean function, the classical method is applying Shannon decomposition.
Intuitively, BDD decomposes Boolean functions with one variable at a time. Therefore, the canonicity of BDD is determined wrt a specific ordering of variables. Crucially, propositional equivalence between two propositional formulas can be checked in constant time once both formulas are converted into BDDs.
The more recent Sentential Decision Diagrams (SDDs)~\cite{Darwiche11} 
utilize a more general decomposition technique that decomposes Boolean functions with a set of variables at each round.
Let $f(\Y \cup \X)$ be a Boolean function over variables $\Y \cup \X$, where $\Y,\X$ are disjoint. Given an $(\Y,\X)$-partition, 
the SDD of $f$, with respect to the $(\Y,\X)$-partition, can be written as $\bigvee^n_{i=1}[\sddprime_i(\Y)\wedge \sddsub_i(\X)]$. 
In particular, besides that all the primes are disjoint and covering, i.e., $\sddprime_i \wedge \sddprime_j=\false$ for $i\neq j$, and $\bigvee^n_{i=1}\sddprime_i = true$, SDD also guarantees that all the subs are compressed, i.e., $\sddsub_i(\X) \neq \sddsub_j(\X)$ for $i \neq j$. Hence, the canonicity of SDDs is determined wrt a specific partition of variables. 

\noindent
\textbf{LTL$_f$ Synthesis.}  
Let $\varphi$ be an \LTLf formula over
$\P=\X \cup \Y$, and $\X, \Y$ are two disjoint sets of propositional variables controlled by the \emph{environment} and the \emph{agent}, respectively. 
%
%
The \emph{synthesis} problem $(\varphi, \X, \Y)$ consists in computing a strategy $g: (2^{\X})^* \rightarrow 2^{\Y}$, such that for an arbitrary infinite sequence $\lambda = X_0,X_1,\ldots\in (2^{\X})^{\omega}$, we can find $k\geq 0$ such that $\rho^k \models \varphi$, where $\rho^k = (X_0\cup g(\epsilon)), (X_1\cup g(X_0)),\ldots , (X_k\cup g(X_0,X_1,\ldots,X_{k-1}))$. If such a strategy does not exist, then $\varphi$ is unrealizable.
%
\LTLf synthesis can be solved by reducing to an adversarial reachability game on the corresponding Deterministic Finite Automaton~(\DFA)~\cite{DegVa15}.  Hence, a strategy can also be represented as a finite-state controller $g: \S \mapsto 2^{\Y}$, where $\S$ denotes the state space of the \DFA.

\let\oldalglinenumber\alglinenumber
\renewcommand{\alglinenumber}[1]{\scriptsize #1:}
\begin{figure}[t]
\vspace{-0.3cm}
\begin{algorithm}[H]
\scriptsize
\caption{\scriptsize SDD-based Forward \LTLf Synthesis [De Giacomo \emph{et al.}, 2022]}
\label[algorithm]{alg:sdd-synthesis}
	\begin{algorithmic}[1]
	
	\Function{Synthesis}{$\varphi$} \Return $\mathsf{strategy}$
	    \If{\textsc{IsAccepting}($\varphi$)}
	        \State\textsc{AddToStrategy}$(\varphi, \mathsf{true})$
	        \State\Return \textsc{GetStrategy}()
	    \EndIf
	    \State \textsc{InitialGraph}($\varphi$) 
	    \State $n\coloneqq$ \textsc{GetGraphRoot}() 
	    \State $\mathsf{found} \coloneqq$ \textsc{Search}($n$, $\emptyset$)
	    \If{$\mathsf{found}$} \Return \textsc{GetStrategy}()
	    \EndIf
	    \State\Return \textsc{EmptyStrategy}() {\small{$\triangleright$ \textit{$\varphi$ is unrealizable}}}
	\EndFunction
	\Function{Search}{$n, \pt$} \Return $\mathsf{True}/\mathsf{False}$
	    \If{\textsc{IsSuccessNode}($n$)} \Return $\mathsf{True}$
	    \EndIf
	    \If{\textsc{IsFailureNode}($n$)} \Return $\mathsf{False}$
	    \EndIf
	    \If{\textsc{InPath}($n, \pt$)}\label{line:in-path} {\small{$\triangleright$ \textit{We found a loop}}}
	    \State \textsc{TagLoop}($n$) \Return $\mathsf{False}$
	    \EndIf
	    \State $\psi \coloneqq$\textsc{FormulaOfNode}($n$)
	    \If{\textsc{IsAccepting}($\psi$)}
	        \State \textsc{TagSuccessNode}($n$)
	        \State \textsc{AddToStrategy}$(\psi, \mathsf{true})$ 
	        \State\Return $\mathsf{True}$
	    \EndIf
	    \State \expand($n$) {\small{$\triangleright$ \textit{Uses SDD to partition $\psi$ wrt $\Y$ and $\X$}}}\label{line:expand}
	    \For{$(act,$ {\small{\andnode}}$)\in $\textsc{GetOrArcs}($n$)}\label{line:get-or-arcs}
	        \For{$(resp, succ)\in $\textsc{GetAndArcs}({\small{\andnode}})}\label{line:get-and-arcs}
	        \State $\mathsf{found} \coloneqq$\textsc{Search}($succ, [\pt|n]$)
	        \If{$\neg \mathsf{found}$} \textbf{Break}\label{line:short-circuit-negative}
	        \EndIf
	    \EndFor
	        \If{$\mathsf{found}$} \label{line:short-circuit-positive}
	        \State \textsc{TagSuccessNode}($n$)
	        \State \textsc{AddToStrategy}$(\psi, act)$
	        \If{{\textsc{IsTagLoop}($n$)}} 
	        \State {\textsc{BackProp}($n$)}\label{line:backprop}
	        \EndIf
	        \State\Return $\mathsf{True}$
	        \EndIf
	    \EndFor
	    \State \textsc{TagFailureNode}($n$)
	    \State\Return $\mathsf{False}$
	\EndFunction
	\end{algorithmic}
\end{algorithm}
\vspace{-0.8cm}
\end{figure}
\renewcommand{\alglinenumber}[1]{\oldalglinenumber{#1}}

\section{Forward \LTLf Synthesis}
\label{sec:preliminaries:ltlf-syn-as-and-or}
Two recent papers \cite{XiaoLZSPV21,GiacomoFLVX022} proposed a \emph{forward} approach to solve the problem of \LTLf synthesis. 
Their implementations are called, respectively, \ltlfsyn and \cynthia.
The idea is to build the \DFA on-the-fly, 
while performing an adversarial forward search towards the final states, by considering the \DFA as a sort of AND-OR graph \cite{Nilsson71}.
Therefore, a winning strategy might be found before constructing the whole \DFA.
Since the implementation of the former approach (\ltlfsyn) has been considered superseded by the latter (\cynthia) in terms of performance, here we focus on \cynthia, although our arguments can be considered more general and not just applicable to specific techniques.

The state-of-the-art forward technique~\cite{GiacomoFLVX022}, implemented in the tool \cynthia, is described by the pseudocode in~\Cref{alg:sdd-synthesis}.
The algorithm is basically a top-down, depth-first traversal of the
AND-OR graph induced by the on-the-fly \DFA construction, proceeding forward from the initial state, and excluding strategies that lead to loops.
The forward-based generation of the AND-OR graph is based on formula progression and on an abstract \expand function~(\Cref{line:expand}) that, taken in input a search node $n$, it produces the next available actions and successor states.
The presence of loops must be carefully handled; when a loop is detected at node $n$, the procedure returns false, temporarily considering $n$ as a failure node. Note that node $n$ is not tagged as failure, since it is unknown whether all the or-arcs of $n$ are explored. If later during the search $n$ is discovered as a success node, such information must be propagated from $n$ backwards to the ancestor nodes of $n$.
It should be noted that, in
a forward search on an AND-OR graph, it is critical to handle loops with the assistance of this backward propagation,
implemented in \textsc{BackProp}~(\Cref{line:backprop}), as illustrated in~\cite{NoteDG}. 
%
The realization of the abstract \expand function was based on SDDs. They have been used for two subtasks: \myi \textbf{state-equivalence checking}, i.e. checking whether two states are equivalent, and \myii \textbf{search node expansion}, i.e. identifying the next AND-OR arcs. The experimental evaluation of 
Their technique is rather impressive, as its implementation \cynthia, outperformed other state-of-the-art tools on challenging benchmarks, e.g. on the Nim benchmark~\cite{BoutonNimAG}. 
For more details on the search algorithm, please refer to the original paper~\cite{GiacomoFLVX022}.

\subsection[Limitations of Previous Works]{Limitations of Previous Works}
\label{sec:limitations}
However, as already acknowledged by the authors, the tool performed poorly on the variant of the \emph{Double Counters} benchmark (cfr. Section 5 of \citealt{GiacomoFLVX022}). We discovered that the main reason is that the search gets stuck with the search node expansion to compute the next agent's and environment's moves, whose number grows exponentially with the scaling parameter $n$, the number of bits of the counters.
In general, we identify at least three factors that hinder the scalability of \cynthia:\\
%
\noindent
\textbf{\myi Disjoint \& covering moves.}
    \ltlfsyn, which naively enumerates all the exponentially-many agent's and environment's moves, has been surpassed by \cynthia. \cynthia is able to branch on disjoint and covering propositional formulas, rather than individual evaluations of agent's and environment's variables, and therefore ending up, most of the times, in a more succinct representation of the next players' moves.
    Nevertheless, for problem instances where the branching factor is very high, the compilation by means of SDDs does not bring much more benefits than exhaustive enumeration, ending up in a huge computational overhead with little usefulness.\\
    \noindent\textbf{\myii No visit before \emph{all} moves are computed.} The search algorithm is constrained by how the identification of the next moves works. That is, the search procedure cannot visit children nodes before all OR arcs, and subsequent AND arcs, have been computed from the current OR-node being expanded. Obviously, a breadth-first search procedure (e.g. AO$^*$~\cite{Nilsson82}) will need to consider all the children of the current search node before proceeding. The point is that, if a search procedure does not need to know in advance all the children of the current node, like \Cref{alg:sdd-synthesis}, then it must be able to do so.\\
    \noindent\textbf{\myiii Equivalence checking and node expansion tighten together.} It is not necessary to tighten together the two tasks of state-equivalence checking and search node expansion. They can be implemented in different ways according to the desired computation trade-offs (e.g. space vs time).

\subsection{A new framework.}
Our aim is to propose a new framework that tries to overcome the above limitations that we consider crucial for a scalable approach.
To do so, we consider a slightly more general version of \Cref{alg:sdd-synthesis}. The generalization is not on the search algorithm being used, but rather on the building blocks that make any AND-OR search algorithm actually suitable for solving \LTLf synthesis in a forward fashion.
In particular, we make a step further from the framework introduced in \cite{GiacomoFLVX022}, which formalizes the search algorithm on top of the \expand function. 
Instead, the two primitive operations that we consider are:
$\EquivalenceCheck(n_1, n_2)$, that checks whether the search nodes $n_1$ and $n_2$ can be considered equivalent wrt the current AND-OR search problem; and $\GetArcs(n)$, that returns an \emph{iterator} of AND-arcs (OR-arcs resp.) of the AND-node (OR-node, resp.) $n$.
In \Cref{alg:sdd-synthesis}, the \EquivalenceCheck procedure is (implicitly) used to check whether a node has been already visited (e.g. see the \InPath function of \Cref{line:in-path}) or to retrieve search tags (e.g. see \textsc{IsSuccessNode}, \textsc{IsFailureNode} and \textsc{IsTagLoop}).
The \GetArcs procedure would be used in place of \textsc{GetOrArcs} and \textsc{GetAndArcs} in Alg. \ref{alg:sdd-synthesis}, at lines \ref{line:get-or-arcs} and \ref{line:get-and-arcs}, respectively.
For the rest of the paper, we consider such \textbf{modified \Cref{alg:sdd-synthesis} as the basis of our techniques.}

The crucial observation is that $\GetArcs(n)$ does not require that the arcs of search node $n$ have already been computed or, in other words, that the node $n$ has been fully expanded (as done by \expand function).
As per specification, $\GetArcs(n)$ is an iterator over the available moves from $n$. 
The concept of iterator is well-known in the computer science community as a way to decouple algorithms from containers~\cite{gamma1995design}.
More interestingly, a special case of iterators, \emph{generators}~\cite{murer1996iteration}, would allow to compute the next players' moves iteratively ``on-demand'', therefore allowing a depth-first search algorithm to visit the next arc returned by the generator even if all arcs have not been computed yet. We will use a generator-based implementation of \GetArcs in \Cref{sec:dpll-based-forward-synthesis}.

In fact, \citeauthor{GiacomoFLVX022}'s approach can be seen as a special case of the proposed framework, in which both \EquivalenceCheck and \GetArcs are implemented using SDDs: two search nodes are equivalent if they point to the same SDD node, and \GetArcs is an iterator that simply scans the children of the root SDD node of $n$. However, this framework can easily overcome the limiting factors identified earlier in this section, namely: \myi computed moves do not have to be disjoint and covering (i.e. different moves that lead to the same successor are allowed, although preferably avoided); \myii if \GetArcs is implemented using a generator-like approach, the visit of a child node can happen far before the computation of all the available moves; and \myiii the two main search subtasks, state-equivalence checking and a search node expansion, are implemented by two potentially decoupled functions (\EquivalenceCheck and \GetArcs, respectively).

\section{Equivalence Checks}
In this section, we describe two equivalence checks that can be used for forward \LTLf synthesis. The first one is a knoweldge-compilation-based equivalence check, that uses BDDs to compile the state formula and achieve constant (propositional) equivalence checking. The second one is a simple and lightweight equivalence check that has never been used in this context and, as we shall see, turns out to be very useful in the experimental evaluation. However, we discuss implications regarding the completeness of the resulting synthesis procedure.

\noindent
\textbf{BDD-based \EquivalenceCheck.} The BDD-based equivalence check is similar to the SDD-based equivalence check performed by \cynthia.
That is, for a search node $n$, we take its associated \LTLf formula $\psi$ with \FormulaOfNode (remember that search node is associated with an \LTLf formula). Then, we compute $\xnf(\psi)$, which is propositionally equivalent to $\psi$. 
$\xnf(\psi)$, by construction, is defined over the set of variables $\Y\cup\X\cup\Z$, where $\Z = \bigcup_{\theta \in \cl(\varphi)} \{z_\alpha | \alpha \in \pa(\xnf(\theta)), \text{$\alpha$ not literal}\}$.
Finally, we get its BDD representation, i.e. $B_{\psi}
\!\!\coloneqq\!\!\BddRepresentation(\xnf(\psi)^p)$. We do these operations both for $n_1$ and $n_2$, yielding $B_{\xnf(\psi_1)}$ and $B_{\xnf(\psi_2)}$. The equivalence check whether the two BDDs point to the same BDD node ($B_{\xnf(\psi_1)}\!\!=\!\!B_{\xnf(\psi_2)}$). If that is the case then it means, thanks to the canonicity property of BDDs, that the associated (propositionalized) formulas are propositionally equivalent.

\begin{lemma}\label{lemma:bdd-eq-check-correctness}
Let $(\varphi, \X , \Y)$ be a \LTLf synthesis problem instance.
The \BddBasedEqCheck procedure for such instance induces a search space for \Cref{alg:sdd-synthesis} with no more than $2^{2^{|\O(\cl(\varphi))|}}$ search nodes.
\end{lemma}
\noindent\textit{Proof.} Any \LTLf formula $\psi$ associated to some search node $n$ of ~\Cref{alg:sdd-synthesis} is such that $\xnf(\psi)^p\in\B^+(\Y\cup\X\cup\Z)$. Since there are at most $2^{|\Y\cup\X\cup\Z|}$ models, thanks to the canonicity property of BDDs, there can be at most $2^{2^{|\Y\cup\X\cup\Z|}}$ propositionally equivalent formulas.
Since $\Y\cup\X\cup\Z = \O(\cl(\varphi))$, we get the claim. \qed

\noindent
We preferred the use of BDDs instead of SDDs
since we do not need the decomposing feature of SDDs, and also because robust and optimized implementations for BDDs already exists, e.g. CUDD \cite{cudd}, with useful features such as dynamic variable reordering.

\noindent
\textbf{Hash-Consing-based \EquivalenceCheck.}
We now consider an equivalence check procedure that is based on \emph{structural equivalence}:
two search nodes $n_1$ and $n_2$ are considered equivalent if their formulas $\psi_1$ and $\psi_2$ have the same syntax tree, i.e.:
$\HashConsingBasedEqCheck(n_1, n_2):=\FormulaOfNode(n_1)=\FormulaOfNode(n_2)$.
%
To make the comparison fast, we can use \emph{hash consing}~\cite{deutsch1973interactive} which is a technique used to share values that are structurally equal. 
Using hash consing, two formulas can be stated as structurally equivalent if they point to the same memory address, achieving constant time equality check. 
%
%
%
Unfortunately, we have the following negative result:
\begin{theorem}
\label{th:hash-consing-non-termination}
 The modified \Cref{alg:sdd-synthesis} with \HashConsingBasedEqCheck for \EquivalenceCheck is sound but not complete for \LTLf synthesis.
\end{theorem}
\noindent\textit{Proof.}
Soundness follows from the soundness of hash-consing based equivalence check.
To disprove completeness, 
consider the synthesis problem with $\varphi=\Box a \lUntil \Diamond b$, $\Y=\{a\}$ and $\X=\{b\}$.
Let $\sigma = \{a\}$, $\varphi_0 = \varphi$ and $\varphi_n = \fp(\varphi_{n-1}, \sigma)$. 
It can be shown by induction on $n$ the following statement: for all $n\ge 1$, we have
$\xnf(\varphi_n)=(((b \wedge \Ftt) \vee \Next\Diamond b) \wedge \Ftt)\vee(\xnf(\varphi_{n-1})\wedge(((a\vee\Gff) \wedge \Wnext\Box a) \vee \Gff)\wedge \Ftt)$.
By correctness of $\fp$, the set of formulas $\varphi_0,\varphi_1.\dots$ are semantically equivalent but structurally different, ending up in a infinite loop, which is undetected by the \HashConsingBasedEqCheck.
\qed

\noindent At the core of the issue is that, by how the formula progression works, there are some cases in which a new structurally different formula can be always produced by some particular sequence of applications of formula progression rules, although propositionally equivalent formulas have been already produced earlier during the search.
Nevertheless, such equivalence check is very computationally cheap and, as we shall see in the experimental section, often it performs better than the BDD-based equivalence check.

To guarantee the termination of this version of the search algorithm, 
we propose the following procedure: given a synthesis problem, first execute the modified \Cref{alg:sdd-synthesis} with \HashConsingBasedEqCheck as equivalence check and some search node expansion procedure.
As soon as, during the execution, the size of the formula of any generated search node becomes greater than a given threshold $t$, then abort the execution and resort to the search algorithm described in \Cref{sec:dpll-based-forward-synthesis}, i.e. \Cref{alg:sdd-synthesis} based on \BddBasedEqCheck and \DpllGetArcs. In other words, if the problem does not present the pathological corner case shown in the proof of \Cref{th:hash-consing-non-termination}, then try to solve it, without getting stuck with onerous BDD-based compilations.
\begin{lemma}
\label{thm:hash-consing-with-threshold}
     The modified \Cref{alg:sdd-synthesis} with \HashConsingBasedEqCheck for \EquivalenceCheck with size formula threshold $t$ always terminates and is correct.
\end{lemma}
\noindent\textit{Proof.} Correctness follows from \cite[Theorem 5]{GiacomoFLVX022}, whereas termination follows from considering that (i) the number of distinct state formulas of size at most $t$ is finite, and (ii) in case the threshold is hit, by the correctness of the BDD-based equivalence check (\Cref{lemma:bdd-eq-check-correctness}). \qed

\noindent
\Cref{thm:hash-consing-with-threshold} says that the threshold guarantees that only a finite number of structurally equivalent formulas can be computed. Empirically, we found that a good threshold that suitably postpones the detection of pathological instances is three times the size of the initial formula: $t=3\cdot |\varphi|$.

\section{DPLL-based Search Node Expansion}
\label{sec:dpll-based-forward-synthesis}
In this section, we describe our main novel approach for search node expansion, that we argue is the key ingredient in achieving state-of-the-art performances for forward \LTLf synthesis, that allows to overcome some limitations of previous works discussed in \Cref{sec:limitations}.
In particular,
%
\GetArcs is implemented using a DPLL-based procedure (\Cref{alg:dpll-based-get-arcs}).
We claim the DPLL-based \GetArcs to be novel and effective for solving our problem, and it is one of the core contributions of the paper.

\noindent
\textbf{DPLL-based \GetArcs.}
The DPLL algorithm \cite{DavisP60,DavisLL62} is a very famous algorithm for deciding the satisfiability of proposition logic formulas in conjunctive normal form (CNF).
Many variants of it have been proposed that work for general non-clausal formulas~\cite{ThiffaultBW04,JainC09}, motivated by the fact that, quite often, conversion of a boolean formula to CNF is both unnecessary and undesirable, e.g. because of loss of structural information and due to the worst-case exponential blow-up of the size of the formula. We agree with this view, and in the following we assume to deal with propositionalized \LTLf formulas in non-clausal form.

We are interested in designing a DPLL-like procedure to identify the next moves and successor nodes from a search node $n$.
Our proposed procedure (\Cref{alg:dpll-based-get-arcs}), like any DPLL procedure, runs by choosing a literal, assigning a truth value to it, simplifying the formula and then recursively applying the same procedure to the simplified formula, until there are no agent or environment variables to branch on. Both the computed set of assignments resulting from the sequence of recursive calls, $ass$ (initialized at~\Cref{line:declaration-of-ass}), and what remains of the formula $\phi = \xnf(\FormulaOfNode(n))^p$ after the chosen literals have been replaced with their assigned truth value, are \emph{yielded}  such that they can be consumed by the caller function (see \Cref{line:final-yield-dpll-get-or-arcs} and \ref{line:final-yield-dpll-get-and-arcs}; the instruction $\mathbf{yield}$ allows a generator to provide a value to the caller). 

Given a search node $n$, $\DpllGetArcs$ returns a generator over pairs $(\movetype, \nodetype)$, where \movetype is a mapping from variables to truth values (the absence of a variable is considered a \emph{don't care}), and \nodetype is a \LTLf formula that, as required by ours and \citeauthor{GiacomoFLVX022}'s search framework, represents a search node (either AND or OR). Depending on whether $n$ is an OR-node or an AND-node, the \DpllGetOrArcs function (\Cref{line:first-call-to-dpll-or-arcs}) or the \DpllGetAndArcs function (\Cref{line:first-call-to-dpll-and-arcs}) is called, respectively.
The \DpllGetOrArcs function takes in input a propositionalization of $\psi$, $\phi$, and the current variables' assignment $ass$.
If there is still some agent variable in $\Y$ to assign (\Cref{line:if-still-some-agent-variable-to-assign}), then we decide the next branching literal $\ell$
(by calling the function \GetBranchingLiteral, \Cref{line:get-agent-branching-literal}), we substitute its truth value to the formula $\phi$, and simplify it by calling the function \Replace (\Cref{line:first-replace-agent}), obtaining $\phi_\ell$. Then, we do the recursive call to \DpllGetOrArcs with the new propositionalized formula $\phi_\ell$ and updated assignment $ass\cup\{\ell\}$, and start generating the next moves with a fixed value for literal $\ell$.
Intuitively, this step represents a transition to another node of the search tree of a DPLL algorithm.
The instruction $\yieldfrom$ allows a generator to forward the generation of results to another generating function.
When the generation terminates, the negated literal $\lnot\ell$ is replaced to the original formula $\phi$, yielding another propositionalized \LTLf formula $\phi_{\lnot\ell}$, and the available moves starting from this branch are generated.
Intuitively, the last step represents the exploration of the opposite branch of the current node of the DPLL search tree, with the branching literal $\ell$ set at the opposite truth value $\lnot \ell$.
Note that in the base case, we return the pair $(ass, \phi^\tf)$, where $ass$ contains all the chosen literals in the current final assignment,
and $\phi^\tf$ is the \LTLf formula that represents the next AND node.
The \DpllGetAndArcs is analogous to \DpllGetOrArcs but for AND nodes; therefore, it aims at finding an assignment of environmentvariables $\X$ rather than of agent variables $\Y$. Another difference with \DpllGetAndArcs is that in the base case, we use the propositional formula $\varPsi$ (the result of the substitutions of chosen literals and the subsequent simplifications)
to compute the next search node formula $\psi'$, using the function $\stripnext$, at \Cref{line:strip-next-get-and-arcs}. 
Note that, at this stage, $\varPsi$ is a propositional formula over $\Z$ state variables only.
By \Cref{prop:xnf}, since $\varPsi = \xnf(\psi)^p|^\P_\sigma$, we have that $\psi' = \stripnext(\varPsi) = \fp(\psi, \sigma)$, i.e. the correct next state.

According to the needs of the search algorithm, the procedure can be run exhaustively, i.e. until all available moves from node $n$ have been produced.
Still, the simplification step can possibly avoid a large part of the naive search space over
$\Y$ and $\X$;
this is an improvement wrt the \ltlfsyn approach, which blindly enumerates all possible assignments.
The simplification step recursively applies the usual propositional simplification rules, e.g. considering the absorbing or neutral boolean values of binary operators.
%
We suggest to simplify the propositional formula to a great extent, but without resorting to any compilations. Instead, we leave the formula in non-clausal form, aiming at eliminating branching variables from the resulting formula. Such variables will be considered as \emph{don't care} in the current assignment.

We argue that such kind of procedures, like the one described in \Cref{alg:dpll-based-get-arcs}, are suitable for our use-case because of their depth-first nature, which implies a low-space requirement, and because of their ``responsive'' nature: a candidate move is proposed in linear time in the number of variables (possibly better thanks to simplifications). 
It is interesting to observe that the full trace of a DPLL execution can seen as a compilation of the propositional theory \cite{HuangD07}.
Note that \Cref{alg:dpll-based-get-arcs} is an abstract specification that can be customized by different realizations of \GetBranchingLiteral and \Replace.

\begin{figure}[t]
\vspace{-0.3cm}
\begin{algorithm}[H]
\footnotesize
\caption{DPLL-based \GetArcs}
\label[algorithm]{alg:dpll-based-get-arcs}
    \begin{algorithmic}[1]
    \Function{\DpllGetArcs}{$n$} \Return $Gen[\movetype, \nodetype]$
        \Let{$\psi$}{$\xnf(\FormulaOfNode(n))$}
        \Let{$ass$}{$\{\}$} \Comment{propositional assignment} \label{line:declaration-of-ass}
        \If{\textsc{IsOrNode($n$)}}
            \State \yieldfrom \DpllGetOrArcs($\psi^p, ass$) \label{line:first-call-to-dpll-or-arcs}
        \Else
            \State \yieldfrom \DpllGetAndArcs($\psi^p, ass$)
            \label{line:first-call-to-dpll-and-arcs}
        \EndIf
    \EndFunction
    \Function{\DpllGetOrArcs}{$\phi, ass$} 
        \Let{$\Y'$}{$\textsc{GetAgentVars}(\phi)$}\label{line:if-still-some-agent-variable-to-assign}
        \If{$\Y' \neq \emptyset$}
             \Let{$\ell$}{$\GetBranchingLiteral(\phi)$}\label{line:get-agent-branching-literal}
             \Let{$\phi_\ell$}{$\Replace(\phi, \ell)$} \label{line:first-replace-agent}
             \State \yieldfrom $\DpllGetOrArcs(\phi_\ell, ass\cup\{\ell\})$\label{line:or-arcs-first-yield}
             \Let{$\phi_{\lnot\ell}$}{$\Replace(\phi, \lnot\ell)$} \label{line:dpll-get-or-arcs-after-first-yield}
             \State \yieldfrom $\DpllGetOrArcs(\phi_{\lnot\ell}, ass \cup \{\lnot\ell\})$ \label{line:or-arcs-second-yield}
        \Else \Comment{No branching on agent variables available}
            \State \yield $(ass, \phi^{\tf})$ \Comment{$\phi^\tf$ is the next AND node} \label{line:final-yield-dpll-get-or-arcs}
        \EndIf
    \EndFunction
    \Function{\DpllGetAndArcs}{$\varPsi, ass$} 
        \Let{$\X'$}{$\textsc{GetEnvVars}(\varPsi)$}
        \If{$\X' \neq \emptyset$}
             \Let{$\ell$}{$\GetBranchingLiteral(\varPsi)$}
             \Let{$\varPsi_\ell$}{$\Replace(\varPsi, \ell)$}
             \State \yieldfrom $\DpllGetAndArcs(\varPsi_\ell, ass \cup \ell)$ \label{line:and-arcs-first-yield}
             \Let{$\varPsi_{\lnot\ell}$}{$\Replace(\varPsi, \lnot\ell)$}
             \State \yieldfrom$\DpllGetAndArcs(\varPsi_{\lnot\ell}, ass\!\cup\!\lnot\ell)$\label{line:and-arcs-second-yield}
        \Else \Comment{No branching on environmentvariables available}
            \Let{$\psi'$}{$\stripnext(\varPsi)$} \label{line:strip-next-get-and-arcs}
            \State \yield $(ass, \psi')$ \Comment{$\psi'$ is the next OR node} \label{line:final-yield-dpll-get-and-arcs}
        \EndIf
    \EndFunction
	\end{algorithmic}
\end{algorithm}
\vspace{-0.8cm}
\end{figure}

\begin{lemma}\label{lemma:dpll-get-arcs-correctness}
Let $(\varphi, \X , \Y)$ be a \LTLf synthesis problem instance. \DpllGetArcs correctly expands the search graph.
\end{lemma}
\noindent\textit{Proof.}
First, we recall some facts of our AND-OR graph construction. Given an OR-node $n$ whose \XNF state formula is $\psi$, for all $Y\in 2^\Y$ possible agent moves, we have an or-arc $(n,Y,m_i)$, where $m_i$ is an AND-node whose associated formula is $\phi_i=(\xnf(\psi)^p|^\Y_Y)^\tf$, and given an AND-node $m_i$, for all $X\in 2^\X$ possible environment moves, we have an and-arc $(m_i,X,n'_{i,j})$, where $n'_{i,j}$ is the successor node with formula 
$\psi'_{i,j} = \fp(\psi,X\!\cup\!Y)$.
%
%
Our claim is the following: there exist an action-response sequence $(n,X,m_i,Y,n'_{i,j})$ iff \myi $(X,m_i)\in\DpllGetOrArcs(n)$, and \myii $(Y,n'_{i,j})\in\DpllGetAndArcs(m_i)$.
%
W.l.o.g., we assume there are no don't care variables; therefore, an assignment $ass_I$ over set of variables $\I$, with $I\in 2^\I$, is such that $p\in I$ iff $p\in ass_I$, and $p\not\in I$ iff $\lnot p\in ass_I$.
%
%
Clearly, $\DpllGetOrArcs$ (resp. $\DpllGetAndArcs$) terminates, because the recursive call is on a formula $\phi_\ell$ (resp. $\varPsi_\ell$) whose set of variables $\Y'$ (resp. $\X'$) is strictly smaller than the one of $\phi$ (resp. $\varPsi$), which means that the same variable is considered only once, and so the final \yield at \Cref{line:final-yield-dpll-get-or-arcs} (resp. \Cref{line:final-yield-dpll-get-and-arcs}) is always executed.
($\Rightarrow$) 
Given an or-arc $(Y,m_i)$, by construction of \DpllGetOrArcs, there exist an execution path where for all $y\in Y$,
$y$ is substituted in $\xnf(\psi^p)$ with $\top$, and for all $y'\not\in Y$, $y'$ is substituted with $\bot$ (via \Replace). 
This must be the case because, by construction, for each variables, both alternatives are considered (Lines \ref{line:or-arcs-first-yield}, \ref{line:or-arcs-second-yield}).
Moreover, on such path, the yielded agent move is $(ass_Y,\xnf(\psi)^p|^\Y_Y)^\tf$, since $ass_Y$ by construction is populated according to the choices on the variable assignments, and because the substitution is a commutative operation.
Analogous reasoning holds for the and-arc $(X,n'_{i,j})$ from $m_i$, with the difference that 
the and-arc generated by \DpllGetAndArcs is $(ass_X,\psi'_{i,j})$, where $\psi'_{i,j}=\stripnext(\varPsi)$ (\Cref{line:strip-next-get-and-arcs}), and $\varPsi=\phi_i|^\X_X=\xnf(\psi)^p|^\P_{X\cup Y}$; therefore, by \Cref{prop:xnf}, $\psi'_{i,j}=\fp(\psi,X\cup Y)$, as desired.
($\Leftarrow$)
If the or-arc $(Y,m_i)$ from $n$ is generated by \DpllGetOrArcs, then as shown earlier, by construction, the state formula associated to $m_i$ is $\phi_i$, which by definition of our AND-OR graph is a valid or-arc from $n$.
Similarly, the and-arc $(X,n'_{i,j})$ from $m_i$ generated by $\DpllGetAndArcs$ is such that the state formula associated to $n'_{i,j}$ is $\stripnext(\xnf(\psi^p)|^\P_{X\cup Y})$, by the reasoning above; therefore, by \Cref{prop:xnf}, $n'_{i,j}$ is a valid successor of $m_i$ in the AND-OR graph. \qed
%
%
%
%
%


\begin{theorem}
Modified \Cref{alg:sdd-synthesis} with \BddBasedEqCheck for state-equivalence checking and \Cref{alg:dpll-based-get-arcs} for search node expansion is correct and always terminates.
\end{theorem}
\noindent\textit{Proof.}
Termination follows from \Cref{lemma:bdd-eq-check-correctness} and Thm. 4 of \cite{GiacomoFLVX022}. Correctness follows from \Cref{lemma:bdd-eq-check-correctness},~\ref{lemma:dpll-get-arcs-correctness}, and Thm. 5 of \cite{GiacomoFLVX022}. \qed

\begin{theorem}
Modified \Cref{alg:sdd-synthesis} with \HashConsingBasedEqCheck and formula size threshold $t$ for state-equivalence checking and \Cref{alg:dpll-based-get-arcs} for search node expansion is correct and always terminates.
\end{theorem}
\noindent\textit{Proof.}
Termination follows from \Cref{thm:hash-consing-with-threshold} and Thm. 4 of \cite{GiacomoFLVX022}. Correctness follows from \Cref{thm:hash-consing-with-threshold},~\ref{lemma:dpll-get-arcs-correctness}, and Thm. 5 of \cite{GiacomoFLVX022}.\qed

%

\begin{example}
\label{example:dpll-running-example}
%
Consider the \LTLf formula $\varphi=a\lUntil b \wedge \Diamond true$, and let $\X=\{a\}$ and $\Y=\{b\}$.
The \XNF form is $\psi=\xnf(\varphi) = ((b\wedge\Diamond true) \vee (a \wedge \Next(a\lUntil b)) \wedge \Diamond true)$.
We start by calling $\DpllGetOrArcs(\psi^p,\{\})$. 
Let $\lnot a=\GetBranchingLiteral(\psi^p)$ be the first chosen literal.
We replace the chosen literal $\lnot a$ in $\psi^p$ (that is, replacing $a$ with $\mathit{false}$), and by simplifying we get $\phi_{\lnot a}= b\wedge \Diamond\mathit{true}$. 
$\psi_{\lnot a}=\phi_{\lnot a}^\tf$ is the \LTLf formula denoting the next AND node.
Now, $\psi_{\lnot a}$ is the next AND node, and the arc $(\{\lnot a\}, \psi_{\lnot a}\})$ is proposed to the search algorithm. This node is then expanded by $\DpllGetAndArcs(\psi_{\lnot a}^p, \{\})$.
Following analogous steps but for environment variables, we get two children nodes $\psi_{\lnot a, \lnot b} = \stripnext(\ffalse) = \ffalse$, and $\psi_{\lnot a, b} = \stripnext(\Diamond(true)) = \ttrue$.
Since $\ffalse$ is a failure node, also the parent AND node is a failure node; therefore, the search algorithm tries the system move $a$.
The execution ``resumes'' from \Cref{line:dpll-get-or-arcs-after-first-yield}, we generate $\psi_{a}=(b\wedge\Diamond(true))\vee \Next(a\lUntil b)$,
and we get the arc $(\psi_{a},\{a\})$.
The OR nodes generated from $\psi_a$ are $\psi_{a,\lnot b} = \stripnext(\Next(a\lUntil b))=(a\lUntil b)\wedge\Diamond\mathit{true}$, and $\psi_{a,b}=\ttrue$ (after some simplification). 
However, note that $\psi_{a,\lnot b}=\varphi$, which we already visited in the current search path.
Therefore, despite children node $\psi_{a,b}$ is a success node, also this AND node $\psi_a$ is a failure node, and so the initial state $\varphi$.
The specification is unrealizable (as expected, since the environment controls the acceptance condition of the $\lUntil$ operator via the variable $b$.
The AND-OR search graph is shown in \Cref{fig:example-1}.
\end{example}

\begin{figure}[t]
\centering
\begin{tikzpicture}
[shorten >=1pt,on grid,auto,transform shape,
scale=0.8,
  or/.style={circle, draw, minimum size=1cm},
  and/.style={rectangle, draw, minimum size=1cm},
  >=stealth]
\node[or] (phi) {$\varphi$};

\node[and, below left=2cm and 1cm of phi] (and1) {$\psi_{\lnot a}$};
\node[or, below= 2cm of and1] (or11) {$\ffalse$};

\node[and, below right=2cm and 1cm of phi] (and2) {$\psi_{a}$};

\node[state,accepting, below=2cm of and2] (or12) {$\ttrue$};

\draw[->] (phi) edge node[midway, above left] {$\lnot a$} (and1);
\draw[->] (phi) edge node[midway, above right] {$a$} (and2);
\draw[->] (and1) edge node[midway, left] {$\lnot b$} (or11);
\draw[->] (and1) edge node[midway, right] {$b$} (or12);
\draw[->] (and2) edge node[midway, left] {$b$} (or12);
\draw[->] (and2) edge[bend right=60,out=180,radius=1cm,looseness=2] node[midway, right] {$\lnot b$} (phi);
\end{tikzpicture}
\caption{Search graph for \Cref{example:dpll-running-example}, with $\varphi=(a\lUntil b)\wedge\Diamond\mathit{true}$.}
\label{fig:example-1}
\end{figure}

\section{Implementation and Experiments}

\begin{figure*}
    \centering
    \begin{subfigure}[b]{.33\linewidth}
    \centering
    \includegraphics[width=1\linewidth]{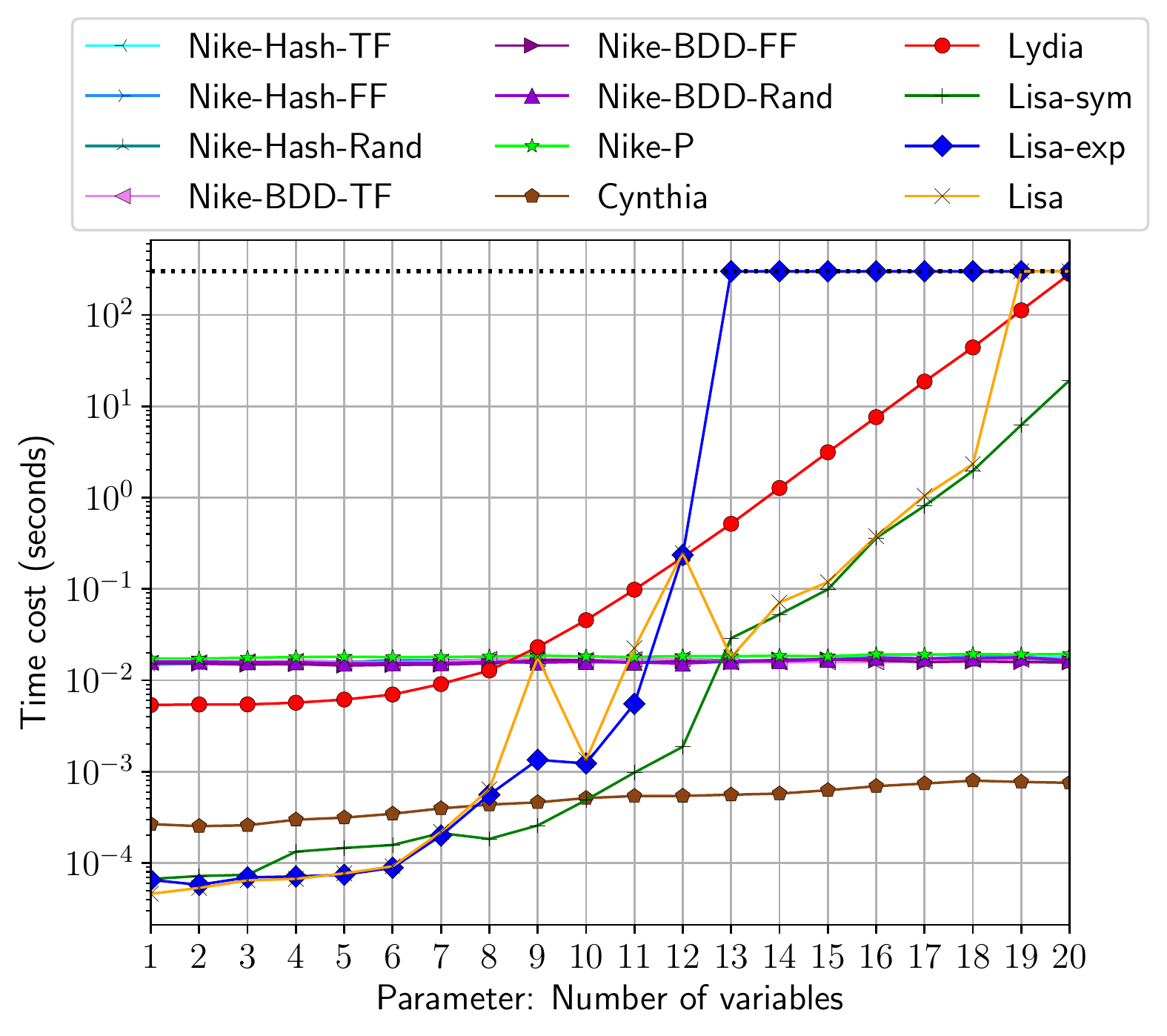}
	\caption{$GF$-pattern.} 
	\label{fig:Gfand}
    \end{subfigure}
    \begin{subfigure}[b]{.33\linewidth}
    \centering
	\includegraphics[width=1\linewidth]{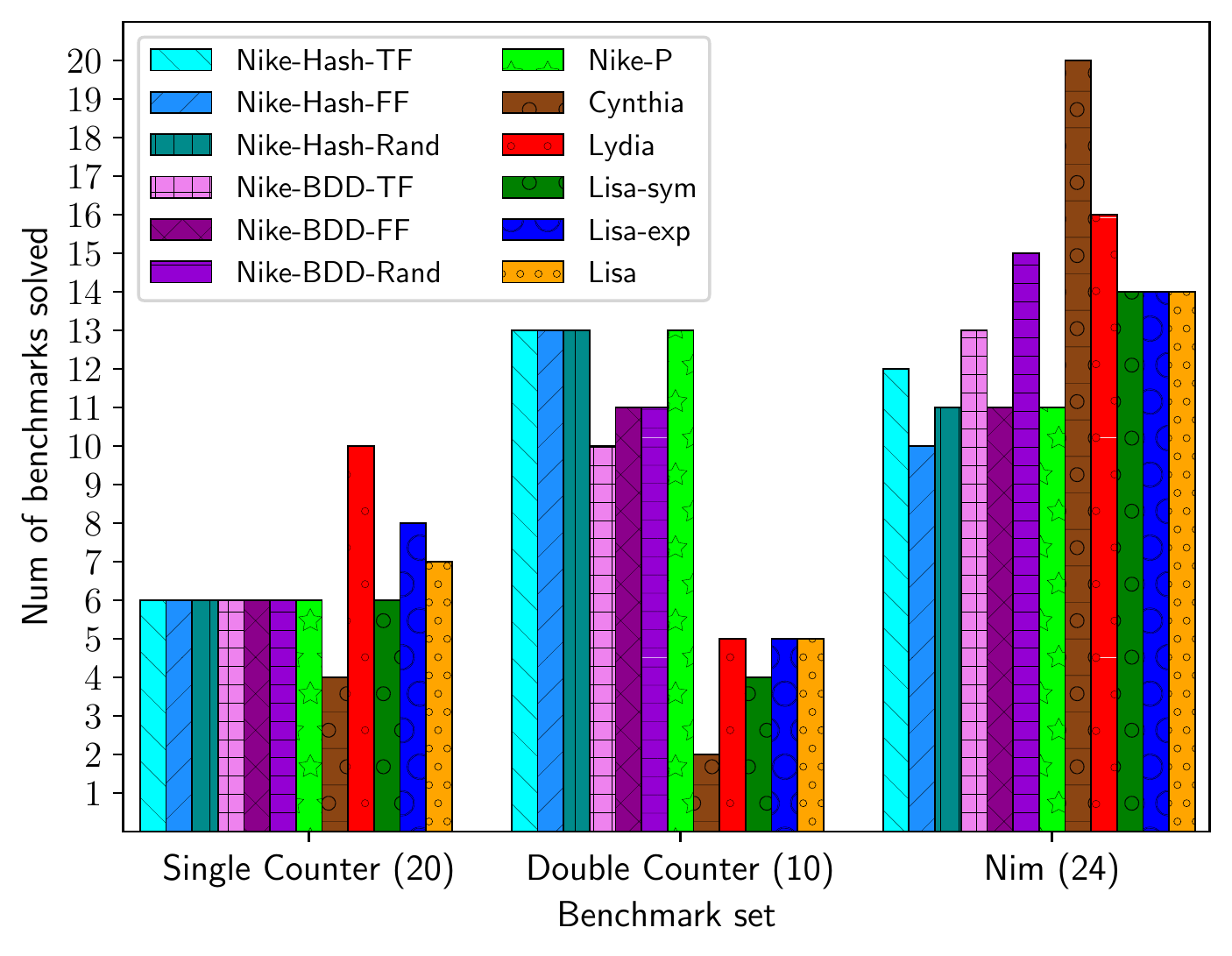}
	\caption{\emph{Two-player-Games}.}
	\label{fig:game}
    \end{subfigure}
    \begin{subfigure}[b]{.33\linewidth}
    \centering
	\includegraphics[width=1\columnwidth]{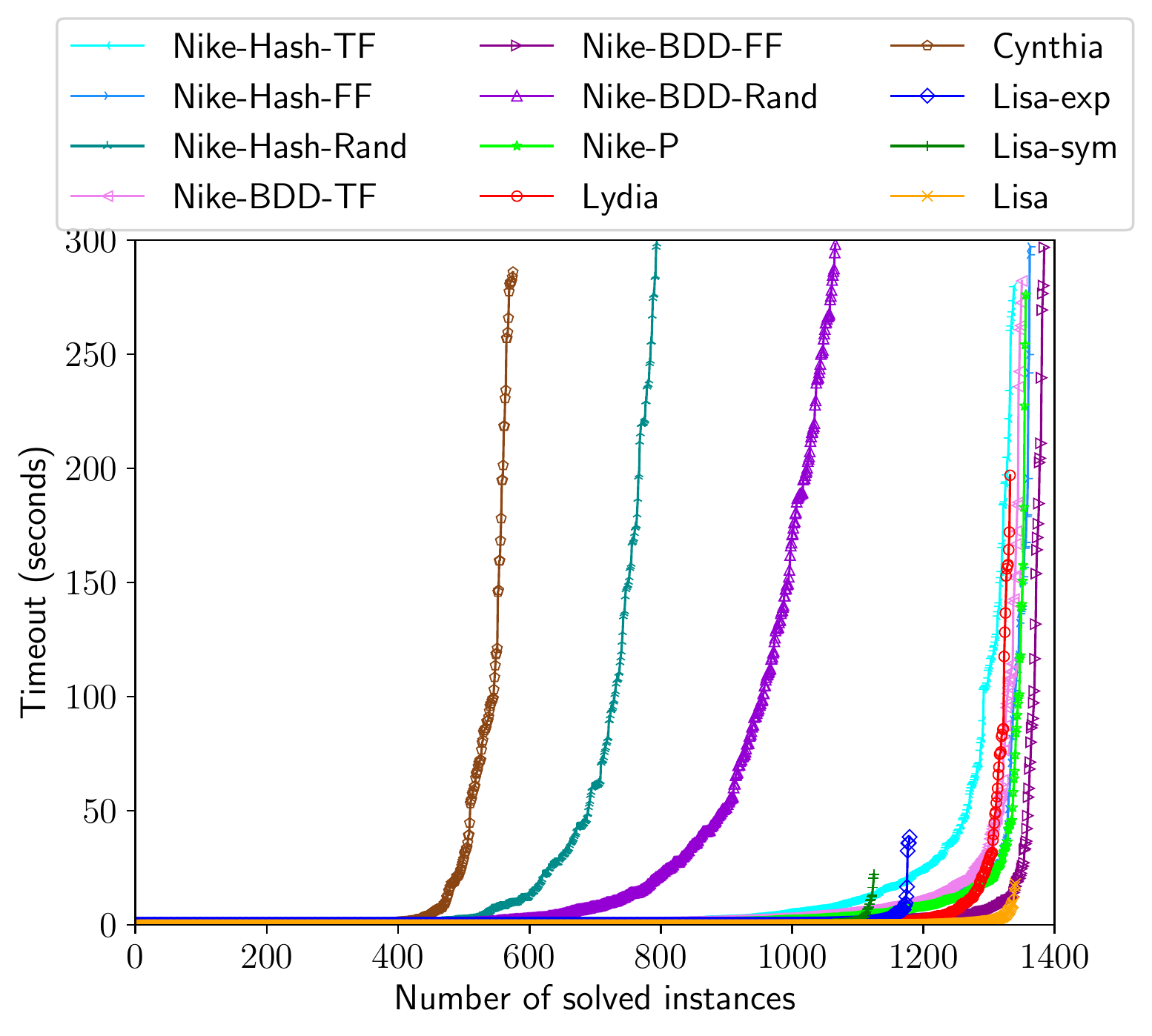}
	\caption{\emph{Random}.}
	\label{fig:random}
    \end{subfigure}
    \caption{Comparison results on all benchmarks.}
\end{figure*}

We implemented the presented synthesis methods in a tool called \nike, which resulted the winner in the \LTLf Realizability Track of SYNTCOMP 2023. \nike is an open-source tool implemented in C++11\footnote{\url{https://github.com/marcofavorito/nike}}.
It uses \syfco to parse the synthesis problems described in TLSF format~\cite{jacobs2023temporal} to obtain the \LTLf specification and the partition of agent/environment propositions. \nike integrates the preprocessing techniques presented in~\cite{XiaoLZSPV21} to perform one-step realizability/unrealizability checks, which is implemented using the BDD library CUDD \cite{cudd}, at the beginning of the synthesis procedure. 
If neither one-step check succeeds, the AND-OR search begins.
Since the procedure is correct and terminates, either the search procedure does not find a winning strategy, in which case the answer to the \LTLf synthesis problem is ``unrealizable'', or a winning strategy is found, and therefore the outcome is ``realizable''.
We use n-ary trees with hash-consing for representing the \LTLf formulas and performing the hash-based state-equivalence checking.
The CUDD library is used for the BDD-based state-equivalence checking.
\nike, as \cynthia and \ltlfsyn, applies some optimizations to speed up the synthesis procedure. First, when visiting an OR-node $n$ for the first time, we perform the pre-processing techniques described in~\cite{XiaoLZSPV21}. More specifically, we check: \myi there exists a one-step strategy that reaches accepting states from $n$, then $n$ is tagged as success; or \myii there does not exist an agent move that can avoid sink state (a non-accepting state only going back to itself) from $n$, then $n$ is tagged as failure.
\nike can run in two modes: using BDD-based state-equivalence checking (\bddbasedeqcheck), and hash-consing-based state-equivalence checking (\hashconsingbasedeqcheck). In the DPLL-based search node expansion, we considered variables in alphabetical order, and we combined them with three simple branching strategies:
\emph{True-first} (\truefirst) that first sets variables to true, \emph{False-first} (\falsefirst) that first sets variables to false; and \emph{Random} (\randomstrategy) that sets variables at random. This yields six combinations of \nike that we included in these experiments. We also include a parallel version of \nike, \nikeparallel, that runs in \hashconsingbasedeqcheck mode all the three branching strategies in parallel.


\noindent
\textbf{Experimental Methodology.}
We evaluated the efficiency of all variants of \nike, by comparing against the following tools: \lisa~\cite{lisa} and Lydia~\cite{GiacomoF21} are state-of-the-art backward \LTLf synthesis approaches. Both tools compute the complete \DFA first, and then solve an adversarial reachability game following the symbolic backward computation technique described in \cite{ZhuTLPV17}. 
 We excluded \ltlfsyn from the comparison since it was already superseded by \cynthia.

\noindent
\textbf{Experiment Setup.}
Experiments were run on a VM instance on Google Cloud, type \texttt{c2-standard-4}, endowed with Intel(R) Xeon(R) CPU 3.10GHz, 4 logical CPU threads,
16 GB of memory and 300 seconds of time limit.
The correctness of \nike was empirically verified by comparing the results with those from all baseline tools. No inconsistency was found.

\noindent
\textbf{Benchmarks.}
We collected 1494 \LTLf synthesis
instances from literature:
20 unrealizable \emph{$GF$-pattern} and 20 realizable \emph{$U$-pattern} instances, of the form
$GF(n)\!\!\!\!=\!\!\!\!\Box(p_1)\!\!\wedge\!\!\Diamond(q_2)\!\!\wedge\!\!\cdots\!\!\wedge\!\!\Diamond q_n)$ and 
$U(n)\!=\!p_1\lUntil(p_2\lUntil(\dots p_{n-1}\lUntil p_n))$, respectively \cite{DBLP:conf/spin/RozierV07,DBLP:conf/spin/GeldenhuysH06}; 54 Two-player-Games instances \cite{TabajaraV19,lisa}: 
\emph{Single-Counter}, 
where the agent stores an $n$-bit counter (where $n$ is the scaling parameter) which it must increment upon a signal by the environment. The agent wins if the counter eventually overflows to $0$;
\emph{Double-Counter}
is similar to the \emph{Single-Counter} one, except that in this case there are two $n$-bit counters, one incremented by the environment and another by the agent. The goal of the agent is for its counter to eventually catch up with the environment's counter;
and \emph{Nim} is a generalized version of the game of Nim with $n$ heaps of $m$ tokens each \cite{BoutonNimAG}. 
Finally, we considered 1400 \emph{Random} instances, of which 400 are from \cite{ZhuTLPV17} and 1000 from \cite{GiacomoF21}, constructed from \LTL synthesis datasets \emph{Lily} and \emph{Load Balancer} \cite{DBLP:conf/tacas/Ehlers11}.
%


\noindent
\textbf{Analysis.}
\Cref{fig:Gfand} shows the running time of each tool on every instance of the \emph{GF-}pattern dataset.
Across these instances, we can observe that all variants of \nike solve instances very quickly, thanks to the pre-processing techniques. This is done with much less time comparing to backward approaches, represented by \lisa and \lydia, simply because these tools do not have such optimizations. \cynthia solved it in less time, but we attribute this to the set up time of the CUDD BDD manager that worsens the performances. Nevertheless, this amounts to a negligible time cost difference of $\ll 1$ second. Results are similar for the \emph{U-}pattern dataset, shown in the supplementary material. 
On the \emph{Two-player-Games} benchmarks, see~\Cref{fig:game}, we observe that \nike variants dominate all other tools on the \emph{Double-Counter} instances, while competing with backward approaches on the other instances. On \emph{Nim}, \cynthia is the best performing tool, but on the other benchmarks \nike shows to be better.
The \nike-\bddbasedeqcheck combinations performs slightly worse on Double Counter than the \nike-\hashconsingbasedeqcheck combinations.
On the \emph{Random} benchmarks, all variants of \nike, except the ones using \randomstrategy, are competitive with state-of-the-art backward approaches, and far better than \cynthia.

It is clear from the plots that \nike, in general, shows
an overall better performance than \cynthia, illustrating the efficiency and better scalability of our approach. In particular,
there is a notable outperformance of \cynthia on the \emph{Double-Counter} and in the \emph{Random} instances.
We attribute this to the ability of \nike to not be stuck with compilation processes that can easily become intractable, both on hand-designed datasets like \emph{Double-Counter}, and in randomly generated intractable cases.
On the \emph{Nim} benchmark, our tool does not perform as good as the others, but its performance are still competitive, especially in the variant \nike-\bddbasedeqcheck-\falsefirst.
%
%
This is because the Nim formulas were manageable enough for SDD compilation (\cynthia) and for DFA construction (\lydia/\lisa), whereas the blind branching strategies of \nike were not effective in this case, as most of the time is spent on generating successors that have been already visited.
The worse performance of the \randomstrategy branching strategy on the \emph{Random} benchmark can be explained by the fact that the \truefirst and the \falsefirst strategies might exploit a particular problem structure of these instances, that allow to easily arrive at success nodes or failure nodes, and frees the algorithm to explore more moves thanks to the short-circuit evaluation of the search outcome (see Lines~\ref{line:short-circuit-negative} and \ref{line:short-circuit-positive} in the modified \Cref{alg:sdd-synthesis}).
The best configuration is \nike-\bddbasedeqcheck-\falsefirst, which suggests that for this benchmark the state compilation is not too hard and the canonicity of the representation helps to prevent the revisit of propositionally-equivalent states.

Overall, despite the simplicity of the DPLL-based expansion, performances are very surprising with respect to backward approaches; this suggests that our approach is very promising and worth of future research.

\section{Conclusions}
We proposed the best forward search \LTLf synthesis approach so far, and the first that is truly competitive with the considered state-of-the-art tools based on backward computation (as in the \emph{Random} benchmark).
Our implementation ranked first in the \LTLf Realizability Track of the 2023 edition of SYNTCOMP.
We think this work sets the foundations for a new family of forward \LTLf synthesis algorithms, and opens several research avenues for investigating effective branching heuristics~\cite{Silva99} for the DPLL-based search graph expansion (e.g. non-chronological backtracking), or better termination strategies for searching with hash-consing-based state-equivalence checking.

\section*{Acknowledgements}
This research has been partially supported by the ERC-ADG WhiteMech (No. 834228). The author  would like to thank Giuseppe De Giacomo and Antonio Di Stasio for their constructive feedback on an earlier draft of this paper.

\bibliographystyle{kr}
\bibliography{kr-sample}

\begin{thebibliography}{}

\bibitem[\protect\citeauthoryear{Bacchus and Kabanza}{1998}]{BacchusK98}
Bacchus, F., and Kabanza, F.
\newblock 1998.
\newblock Planning for temporally extended goals.
\newblock {\em Ann. Math. Artif. Intell.} 22(1-2).

\bibitem[\protect\citeauthoryear{Baier and Katoen}{2008}]{BaierKatoenBook08}
Baier, C., and Katoen, J.
\newblock 2008.
\newblock {\em Principles of model checking}.

\bibitem[\protect\citeauthoryear{Bansal \bgroup et al\mbox.\egroup
  }{2020a}]{lisa}
Bansal, S.; Li, Y.; {M. Tabajara}, L.; and {Y. Vardi}, M.
\newblock 2020a.
\newblock Hybrid compositional reasoning for reactive synthesis from
  finite-horizon specifications.
\newblock In {\em {AAAI}}.

\bibitem[\protect\citeauthoryear{Bansal \bgroup et al\mbox.\egroup
  }{2020b}]{BLTV}
Bansal, S.; Li, Y.; {M. Tabajara}, L.; and {Y. Vardi}, M.
\newblock 2020b.
\newblock {Hybrid Compositional Reasoning for Reactive Synthesis from
  Finite-Horizon Specifications.}
\newblock In {\em {AAAI}}.

\bibitem[\protect\citeauthoryear{Bertoli \bgroup et al\mbox.\egroup
  }{2006}]{BertoliCRT06}
Bertoli, P.; Cimatti, A.; Roveri, M.; and Traverso, P.
\newblock 2006.
\newblock Strong planning under partial observability.
\newblock {\em Artif. Intell.} 170(4-5).

\bibitem[\protect\citeauthoryear{Bouton}{1901}]{BoutonNimAG}
Bouton, C.~L.
\newblock 1901.
\newblock Nim, a game with a complete mathematical theory.
\newblock {\em Annals of Mathematics} 3.

\bibitem[\protect\citeauthoryear{{Brafman}, {De Giacomo}, and
  Patrizi}{2018}]{BrafmanGP18}
{Brafman}, R.~I.; {De Giacomo}, G.; and Patrizi, F.
\newblock 2018.
\newblock {LTL}$_f$/{LDL}$_f$ non-markovian rewards.
\newblock In {\em {AAAI}}.

\bibitem[\protect\citeauthoryear{Bryant}{1992}]{Bryant92}
Bryant, R.~E.
\newblock 1992.
\newblock {Symbolic Boolean Manipulation with Ordered Binary-Decision
  Diagrams}.
\newblock {\em {ACM} Comput. Surv.} 24(3).

\bibitem[\protect\citeauthoryear{Camacho and {A. McIlraith}}{2019}]{CamachoM19}
Camacho, A., and {A. McIlraith}, S.
\newblock 2019.
\newblock Strong fully observable non-deterministic planning with {LTL} and
  {LTL}$_f$ goals.
\newblock In {\em {IJCAI}}.

\bibitem[\protect\citeauthoryear{Camacho \bgroup et al\mbox.\egroup
  }{2018}]{CamachoBMM18}
Camacho, A.; Baier, J.~A.; Muise, C.~J.; and McIlraith, S.~A.
\newblock 2018.
\newblock {Finite {LTL} Synthesis as Planning}.
\newblock In {\em {ICAPS}}.

\bibitem[\protect\citeauthoryear{Chakrabarti}{1994}]{Chakrabarti94}
Chakrabarti, P.~P.
\newblock 1994.
\newblock Algorithms for searching explicit {AND/OR} graphs and their
  applications to problem reduction search.
\newblock {\em Artif. Intell.} 65(2).

\bibitem[\protect\citeauthoryear{Church}{1963}]{church1963application}
Church, A.
\newblock 1963.
\newblock Application of recursive arithmetic to the problem of circuit
  synthesis.
\newblock {\em Journal of Symbolic Logic} 28(4).

\bibitem[\protect\citeauthoryear{Cimatti \bgroup et al\mbox.\egroup
  }{2003}]{Cimatti03}
Cimatti, A.; Pistore, M.; Roveri, M.; and Traverso, P.
\newblock 2003.
\newblock Weak, strong, and strong cyclic planning via symbolic model checking.
\newblock 1--2(147).

\bibitem[\protect\citeauthoryear{Cimatti, Roveri, and
  Traverso}{1998}]{CimattiRT98}
Cimatti, A.; Roveri, M.; and Traverso, P.
\newblock 1998.
\newblock Strong planning in non-deterministic domains via model checking.
\newblock In {\em {AIPS}}.

\bibitem[\protect\citeauthoryear{Darwiche and Marquis}{2002}]{DarwicheM02}
Darwiche, A., and Marquis, P.
\newblock 2002.
\newblock A knowledge compilation map.
\newblock {\em J. Artif. Intell. Res.} 17:229--264.

\bibitem[\protect\citeauthoryear{Darwiche}{2011}]{Darwiche11}
Darwiche, A.
\newblock 2011.
\newblock {SDD:} {A} new canonical representation of propositional knowledge
  bases.
\newblock In {\em {IJCAI}}.

\bibitem[\protect\citeauthoryear{Davis and Putnam}{1960}]{DavisP60}
Davis, M., and Putnam, H.
\newblock 1960.
\newblock A computing procedure for quantification theory.
\newblock {\em J. {ACM}} 7(3):201--215.

\bibitem[\protect\citeauthoryear{Davis, Logemann, and
  Loveland}{1962}]{DavisLL62}
Davis, M.; Logemann, G.; and Loveland, D.~W.
\newblock 1962.
\newblock A machine program for theorem-proving.
\newblock {\em Commun. {ACM}} 5(7):394--397.

\bibitem[\protect\citeauthoryear{{De Giacomo} and Favorito}{2021}]{GiacomoF21}
{De Giacomo}, G., and Favorito, M.
\newblock 2021.
\newblock Compositional approach to translate {LTL}$_f$/{LDL}$_f$ into
  deterministic finite automata.
\newblock In {\em {ICAPS}}.

\bibitem[\protect\citeauthoryear{{De Giacomo} and Vardi}{2013}]{DegVa13}
{De Giacomo}, G., and Vardi, M.~Y.
\newblock 2013.
\newblock {Linear Temporal Logic and Linear Dynamic Logic on Finite Traces}.
\newblock In {\em {IJCAI}}.

\bibitem[\protect\citeauthoryear{{De Giacomo} and Vardi}{2015}]{DegVa15}
{De Giacomo}, G., and Vardi, M.~Y.
\newblock 2015.
\newblock {Synthesis for {LTL} and {LDL} on Finite Traces}.
\newblock In {\em {IJCAI}}.

\bibitem[\protect\citeauthoryear{{De Giacomo} \bgroup et al\mbox.\egroup
  }{2022}]{GiacomoFLVX022}
{De Giacomo}, G.; Favorito, M.; Li, J.; Vardi, M.~Y.; Xiao, S.; and Zhu, S.
\newblock 2022.
\newblock Ltlf synthesis as {AND-OR} graph search: Knowledge compilation at
  work.
\newblock In {\em {IJCAI}},  2591--2598.
\newblock ijcai.org.

\bibitem[\protect\citeauthoryear{Deutsch}{1973}]{deutsch1973interactive}
Deutsch, L.~P.
\newblock 1973.
\newblock An interactive program verifier.

\bibitem[\protect\citeauthoryear{Ehlers}{2010}]{Unbeast}
Ehlers, R.
\newblock 2010.
\newblock {Symbolic Bounded Synthesis}.
\newblock In {\em {CAV}}.

\bibitem[\protect\citeauthoryear{Ehlers}{2011}]{DBLP:conf/tacas/Ehlers11}
Ehlers, R.
\newblock 2011.
\newblock Unbeast: Symbolic bounded synthesis.
\newblock In {\em {TACAS}}, volume 6605 of {\em Lecture Notes in Computer
  Science},  272--275.
\newblock Springer.

\bibitem[\protect\citeauthoryear{Emerson}{1990}]{Emerson1990TemporalAM}
Emerson, E.~A.
\newblock 1990.
\newblock Temporal and modal logic.
\newblock In {\em Handbook of Theoretical Computer Science}.

\bibitem[\protect\citeauthoryear{Gamma \bgroup et al\mbox.\egroup
  }{1995}]{gamma1995design}
Gamma, E.; Johnson, R.; Helm, R.; Johnson, R.~E.; and Vlissides, J.
\newblock 1995.
\newblock {\em Design patterns: elements of reusable object-oriented software}.
\newblock Pearson Deutschland GmbH.

\bibitem[\protect\citeauthoryear{Geffner and Bonet}{2013}]{GeBo13}
Geffner, H., and Bonet, B.
\newblock 2013.
\newblock {\em A Concise Introduction to Models and Methods for Automated
  Planning}.

\bibitem[\protect\citeauthoryear{Geldenhuys and
  Hansen}{2006}]{DBLP:conf/spin/GeldenhuysH06}
Geldenhuys, J., and Hansen, H.
\newblock 2006.
\newblock Larger automata and less work for {LTL} model checking.
\newblock In {\em {SPIN}}, volume 3925 of {\em Lecture Notes in Computer
  Science},  53--70.
\newblock Springer.

\bibitem[\protect\citeauthoryear{Ghallab, {S. Nau}, and
  Traverso}{2004}]{ghallab2004automatedplanning}
Ghallab, M.; {S. Nau}, D.; and Traverso, P.
\newblock 2004.
\newblock {\em Automated planning - theory and practice}.

\bibitem[\protect\citeauthoryear{Goldman and Boddy}{1996}]{GoBo96}
Goldman, R.~P., and Boddy, M.~S.
\newblock 1996.
\newblock Expressive planning and explicit knowledge.
\newblock In {\em {AIPS}}.

\bibitem[\protect\citeauthoryear{Haslum \bgroup et al\mbox.\egroup
  }{2019}]{2019Haslum}
Haslum, P.; Lipovetzky, N.; Magazzeni, D.; and Muise, C.
\newblock 2019.
\newblock {\em An Introduction to the Planning Domain Definition Language}.

\bibitem[\protect\citeauthoryear{Huang and Darwiche}{2007}]{HuangD07}
Huang, J., and Darwiche, A.
\newblock 2007.
\newblock The language of search.
\newblock {\em J. Artif. Intell. Res.} 29:191--219.

\bibitem[\protect\citeauthoryear{{J. Nilsson}}{1982}]{Nilsson82}
{J. Nilsson}, N.
\newblock 1982.
\newblock {\em Principles of Artificial Intelligence}.

\bibitem[\protect\citeauthoryear{Jacobs, Perez, and
  Schlehuber-Caissier}{2023}]{jacobs2023temporal}
Jacobs, S.; Perez, G.~A.; and Schlehuber-Caissier, P.
\newblock 2023.
\newblock The temporal logic synthesis format tlsf v1.2.

\bibitem[\protect\citeauthoryear{Jain and Clarke}{2009}]{JainC09}
Jain, H., and Clarke, E.~M.
\newblock 2009.
\newblock Efficient {SAT} solving for non-clausal formulas using dpll, graphs,
  and watched cuts.
\newblock In {\em {DAC}},  563--568.
\newblock {ACM}.

\bibitem[\protect\citeauthoryear{Jim{\'{e}}nez and Torras}{2000}]{JimenezT00}
Jim{\'{e}}nez, P., and Torras, C.
\newblock 2000.
\newblock An efficient algorithm for searching implicit {AND/OR} graphs with
  cycles.
\newblock {\em Artif. Intell.} 124(1).

\bibitem[\protect\citeauthoryear{Jobstmann and Bloem}{2006}]{Lily}
Jobstmann, B., and Bloem, R.
\newblock 2006.
\newblock Optimizations for {LTL} synthesis.
\newblock In {\em {FMCAD}}.

\bibitem[\protect\citeauthoryear{Li \bgroup et al\mbox.\egroup
  }{2019}]{LiRPZV19}
Li, J.; Rozier, K.~Y.; Pu, G.; Zhang, Y.; and Vardi, M.~Y.
\newblock 2019.
\newblock Sat-based explicit {LTL}$_f$ satisfiability checking.
\newblock In {\em {AAAI}}.

\bibitem[\protect\citeauthoryear{Mahanti and Bagchi}{1985}]{MahantiB85}
Mahanti, A., and Bagchi, A.
\newblock 1985.
\newblock {AND/OR} graph heuristic search methods.
\newblock {\em J. {ACM}} 32(1).

\bibitem[\protect\citeauthoryear{Mattm{\"{u}}ller \bgroup et al\mbox.\egroup
  }{2010}]{MyND}
Mattm{\"{u}}ller, R.; Ortlieb, M.; Helmert, M.; and Bercher, P.
\newblock 2010.
\newblock Pattern database heuristics for fully observable nondeterministic
  planning.
\newblock In {\em {ICAPS}}.

\bibitem[\protect\citeauthoryear{Mattm{\"{u}}ller}{2013}]{Mattmuller13}
Mattm{\"{u}}ller, R.
\newblock 2013.
\newblock {\em Informed progression search for fully observable
  nondeterministic planning}.
\newblock Ph.D. Dissertation.

\bibitem[\protect\citeauthoryear{Murer \bgroup et al\mbox.\egroup
  }{1996}]{murer1996iteration}
Murer, S.; Omohundro, S.; Stoutamire, D.; and Szyperski, C.
\newblock 1996.
\newblock Iteration abstraction in sather.
\newblock {\em ACM Transactions on Programming Languages and Systems (TOPLAS)}
  18(1):1--15.

\bibitem[\protect\citeauthoryear{{Nilsson}}{1971}]{Nilsson71}
{Nilsson}, N.~J.
\newblock 1971.
\newblock {\em Problem-solving methods in artificial intelligence}.

\bibitem[\protect\citeauthoryear{Pnueli and Rosner}{1989}]{PnueliR89}
Pnueli, A., and Rosner, R.
\newblock 1989.
\newblock {On the Synthesis of a Reactive Module}.
\newblock In {\em {POPL}}.

\bibitem[\protect\citeauthoryear{Pnueli}{1977}]{Pnu77}
Pnueli, A.
\newblock 1977.
\newblock {The temporal logic of programs}.
\newblock In {\em FOCS}.

\bibitem[\protect\citeauthoryear{Reif}{1984}]{Reif84}
Reif, J.~H.
\newblock 1984.
\newblock The complexity of two-player games of incomplete information.
\newblock {\em JCSS} 29(2).

\bibitem[\protect\citeauthoryear{Rintanen}{2004}]{Rintanen04a}
Rintanen, J.
\newblock 2004.
\newblock Complexity of planning with partial observability.
\newblock In {\em {ICAPS}}.

\bibitem[\protect\citeauthoryear{Rozier and
  Vardi}{2007}]{DBLP:conf/spin/RozierV07}
Rozier, K.~Y., and Vardi, M.~Y.
\newblock 2007.
\newblock {LTL} satisfiability checking.
\newblock In {\em {SPIN}}, volume 4595 of {\em Lecture Notes in Computer
  Science},  149--167.
\newblock Springer.

\bibitem[\protect\citeauthoryear{Scutell\`{a}}{1990}]{NoteDG}
Scutell\`{a}, M.~G.
\newblock 1990.
\newblock A note on dowling and gallier's top-down algorithm for propositional
  horn satisfiability.
\newblock {\em J. Log. Program.} 8(3):265–273.

\bibitem[\protect\citeauthoryear{Silva}{1999}]{Silva99}
Silva, J. P.~M.
\newblock 1999.
\newblock The impact of branching heuristics in propositional satisfiability
  algorithms.
\newblock In {\em {EPIA}}, volume 1695 of {\em Lecture Notes in Computer
  Science},  62--74.
\newblock Springer.

\bibitem[\protect\citeauthoryear{Somenzi}{2016}]{cudd}
Somenzi, F.
\newblock 2016.
\newblock {CUDD: CU Decision Diagram Package 3.0.0. Universiy of Colorado at
  Boulder}.

\bibitem[\protect\citeauthoryear{Tabajara and Vardi}{2019}]{TabajaraV19}
Tabajara, L.~M., and Vardi, M.~Y.
\newblock 2019.
\newblock {Partitioning Techniques in LTL$_f$ Synthesis}.
\newblock In {\em {IJCAI}}.

\bibitem[\protect\citeauthoryear{{Thanh To}, Pontelli, and {Cao
  Son}}{2009}]{ToPS09}
{Thanh To}, S.; Pontelli, E.; and {Cao Son}, T.
\newblock 2009.
\newblock A conformant planner with explicit disjunctive representation of
  belief states.
\newblock In {\em {ICAPS}}.

\bibitem[\protect\citeauthoryear{Thiffault, Bacchus, and
  Walsh}{2004}]{ThiffaultBW04}
Thiffault, C.; Bacchus, F.; and Walsh, T.
\newblock 2004.
\newblock Solving non-clausal formulas with {DPLL} search.
\newblock In {\em {SAT}}.

\bibitem[\protect\citeauthoryear{Xiao \bgroup et al\mbox.\egroup
  }{2021}]{XiaoLZSPV21}
Xiao, S.; Li, J.; Zhu, S.; Shi, Y.; Pu, G.; and {Y. Vardi}, M.
\newblock 2021.
\newblock On-the-fly synthesis for {LTL} over finite traces.
\newblock In {\em {AAAI}}.

\bibitem[\protect\citeauthoryear{Zhu \bgroup et al\mbox.\egroup
  }{2017a}]{ZTLPV17}
Zhu, S.; Tabajara, L.~M.; Li, J.; Pu, G.; and Vardi, M.~Y.
\newblock 2017a.
\newblock {Symbolic {LTL$_f$} Synthesis}.
\newblock In {\em {IJCAI}}.

\bibitem[\protect\citeauthoryear{Zhu \bgroup et al\mbox.\egroup
  }{2017b}]{ZhuTLPV17}
Zhu, S.; Tabajara, L.~M.; Li, J.; Pu, G.; and Vardi, M.~Y.
\newblock 2017b.
\newblock {A Symbolic Approach to Safety LTL Synthesis}.
\newblock In {\em HVC}.

\end{thebibliography}

\section{Proofs}

\subsection{Algorithm 1 + \HashConsingBasedEqCheck + \DpllGetArcs}
\begin{figure}[H]
  \centering
  \begin{tikzpicture}[scale=0.7, every node/.style={transform shape}]
    \node[state, initial] (q0) {$q_0$};
    \node[state, accepting, right of=q0] (q1) {$q_1$};
    \path   
      (q0) edge[loop above] node{$\lnot b$} (q0)
      (q0) edge[above] node{$b$} (q1)
      (q1) edge[loop above] node{$true$} (q1);
    \end{tikzpicture}
    \caption{Minimal DFA of $\varphi = \Box a\lUntil \Diamond b$}
    \label{fig:dfa-a-until-b}
\end{figure}

\setcounter{theorem}{3}
\begin{theorem}
\label{th:hash-consing-non-termination}
Modified Algorithm 1 with Equation 1 for \EquivalenceCheck and Algorithm 2 for \GetArcs is sound but not complete for \LTLf synthesis.
\end{theorem}
\begin{proof}
Soundness follows from correctness of \DpllGetArcs, from the soundness of hash-consing based equivalence check, and the correctness of Algorithm 1 (Theorem 5 of \cite{GiacomoFLVX022}).

To disprove completeness, we show there exist a synthesis problem $(\varphi, \X, \Y)$ such that the algorithm does not terminate.
Let $\varphi=\Box a \lUntil \Diamond b$, with $\Y=\{a\}$ and $\X=\{b\}$. The equivalent automaton of $\varphi$ is shown in \Cref{fig:dfa-a-until-b}.
Consider any assignment with $b$ set to false, e.g. $\sigma = \{a\}$. The repeated exploration of the agent-env move pair equivalent to $\sigma$ makes the formula progression to produce ever bigger state formulas, hence making the hash-consing-based equivalence check to return false, although the associated state of the minimal DFA is always the same ($q_0$), see again \Cref{fig:dfa-a-until-b}. 

In particular, we prove by induction the following statement. Let $\varphi_0 = \varphi$ and $\varphi_n = \fp(\varphi_{n-1}, \sigma)$. For all $n\ge 1$, we have:
\begin{align*}
\xnf(\varphi_n)=&(((b \wedge \Ftt) \vee \Next\Diamond b) \wedge \Ftt) \\
& \vee (  \\
& \quad \xnf(\varphi_{n-1}) \\
& \quad \wedge (((a\vee\Gff) \wedge \Wnext\Box a) \vee \Gff) \\
& \quad \wedge \Ftt \\
& )
\end{align*}

Base step $n=1$:
the initial state formula in XNF form $\xnf(\varphi)$ is the following:
\begin{align*}
\xnf(\varphi)=&(((b \wedge \Ftt) \vee \Next\Diamond b) \wedge \Ftt ) \\
&\vee \\
&(((a \vee \Gff) \wedge \Wnext\Box a \wedge \Next(\Box a \lUntil \Diamond b)))
\end{align*}

After applying the transformation to move to the next state, $\stripnext(\xnf(\varphi)^p|_\sigma)$, or, equivalently, $\fp(\varphi, \sigma)$, we get a new \LTLf formula, $\varphi_1$, that in XNF form becomes $\xnf(\varphi_1)$:
\begin{align*}
\xnf(\varphi_1)=&(((b \wedge \Ftt) \vee \Next\Diamond b) \wedge \Ftt) \\
& \vee (  \\
& \quad \xnf(\varphi_0) \\
& \quad \wedge (((a\vee\Gff) \wedge \Wnext\Box a) \vee \Gff) \\
& \quad \wedge \Ftt \\
& )
\end{align*}
Note that the original formula $\xnf(\varphi)$ appears in the formula $\xnf(\varphi_0)$. Therefore, the claim holds.

Inductive step. Let the claim hold for all $i\le n$, we need to prove that the claim holds for $n+1$.
By inductive hypothesis, we have that
\begin{align*}
\xnf(\varphi_n)=&(((b \wedge \Ftt) \vee \Next\Diamond b) \wedge \Ftt) \\
& \vee (  \\
& \quad \xnf(\varphi_{n-1}) \\
& \quad \wedge (((a\vee\Gff) \wedge \Wnext\Box a) \vee \Gff) \\
& \quad \wedge \Ftt \\
& )
\end{align*}

Once applying again the same transformation but for formula $\varphi_{n+1}$, i.e. $\fp(\xnf(\varphi_n), \sigma)$, and then applying the XNF, it can be shown that we obtain the formula $\xnf(\varphi_{n+1})$:

\begin{align*}
\xnf(\varphi_{n+1})=&(((b \wedge \Ftt) \vee \Next\Diamond b) \wedge \Ftt) \\
& \vee (  \\
& \quad \xnf(\varphi_{n}) \\
& \quad \wedge (((a\vee\Gff) \wedge \Wnext\Box a) \vee \Gff) \\
& \quad \wedge \Ftt \\
& )
\end{align*}
note that we have a pattern: the new formula contains as a subformula the formulas computed at the previous steps.

Moreover, it can be shown that the formulas are semantically equivalent, i.e. $\varphi_n \equiv \varphi_{n-1}$ for all $n\ge 1$ (e.g. using any \LTLf-to-\DFA tool, like \lydia), and therefore, the search will loop in the same semantically equivalent state, but on structurally different state formulas, hence without progresses of the search. See the script in \texttt{benchmark/proof-theorem-3.py} in the supplementary material.
\end{proof}

\section{Empirical Evaluations}

\subsection*{Benchmarks}
We collected, in total, 1494 \LTLf synthesis instances from literature, consisting of 3 benchmark families: 40 patterned instances from the \emph{Patterns} benchmark family~\cite{XiaoLZSPV21}, split into the $GF$-pattern and $U$-pattern datasets; 54 instances from the \emph{Two-player-Games} benchmark family~\cite{TabajaraV19,BLTV}, split into Single-Counter, Double-Counters and Nim datasets. Since the formulation there assumes that the environment acts first, the \LTLf instances had to be modified slightly to adapt to our setting, where the agent acts first; 1400 randomly conjuncted instances taken from~\cite{ZTLPV17,GiacomoF21}.

\paragraph{Patterns.} There are 20 unrealizable $GF$-pattern instances, and 20 realizable $U$-pattern instances, constructed in the following ways, respectively.
\begin{align*}
    & GF(n) = G(p_1) \wedge F(q_2) \wedge F(q_3) \wedge \ldots \wedge F(q_n) \\
    & U(n) = p_1 U (p_2 U (\ldots p_{n-1} U p_n))
\end{align*}
More specifically, $G$ stands for $\Box$~(Always), $F$ stands for $\Diamond$~(Eventually), and $U$ stands for $\lUntil$~(Until). The variables in the formulas are roughly equally partitioned into $\X$ and $\Y$ at random. In particular, for $GF$-pattern instances, the first variable $p_1$ is always assigned as environment variable such that all generated instances are guaranteed to be unrealizable. Moreover, for $U$-pattern instances, the last variable $p_n~(n \geq 2)$ is always assigned as agent variable such that all generated instances are guaranteed to be realizable.

\paragraph{Two-player-Games.} 
\emph{Single-Counter} is a simple example where the behavior of the agent is completely determined by the actions of the environment. Therefore, the challenge in this family lies mostly in proving that the specification is realizable. The agent stores an $n$-bit counter (where $n$ is the scaling parameter) which it must increment upon a signal by the environment. The agent wins if the counter eventually overflows to 0. To guarantee that the game is winning for the agent, the specification assumes that the environment will send the increment signal at least once every two timesteps.

\emph{Double-Counter} is similar to the \emph{Single-Counter} one, except that in this case there are two $n$-bit counters, one incremented by the environment and another by the agent. The goal of the agent is for its counter to eventually catch up with the environment’s counter. To guarantee that this is achievable, the specification assumes that the environment cannot increment its counter twice in a row.

\emph{Nim} describes a generalized version of the game of Nim~\cite{BoutonNimAG} with $n$ heaps of $m$ tokens each. The environment and the agent take turns removing any number of tokens from one of the heaps, and the player who removes the last token loses.

\paragraph{Random.} This benchmark family has 1400 instances, from which there are 1000 instances from~\cite{ZTLPV17}, and 400 instances from~\cite{GiacomoF21}. The instances in this benchmark family are constructed from basic cases taken from \LTL synthesis datasets Lily~\cite{Lily} and Load balancer~\cite{Unbeast}. Formally, a random conjunction formula $RC(L)$ has the form: $RC(L) = \bigwedge_{1\leq i\leq L}P_i(v_1,v_2,...,v_k),$
where $L$ is the number of conjuncts, or the length of the formula, and $P_i$ is a randomly selected basic case. Variables $v_1,v_2,\ldots,v_k$ are chosen randomly from a set of $m$ candidate variables. Given $L$ and $m$~(the size of the candidate variable set), we generate a formula $RC(L)$ in the following way:
\begin{enumerate}
    \item 
    Randomly select $L$ basic cases;
    \item
    For each case $\varphi$, substitute every variable $v$ with a random new variable $v'$ chosen from $m$ atomic propositions. If $v$ is an environment-variable in $\varphi$, then $v'$ is also an environment-variable in $RC(L)$. The same applies to the agent-variables.
\end{enumerate}

\section{Plots}
Here we provide the full set of experimental results.

\begin{figure}[H]
\centering
\includegraphics[width=0.8\columnwidth]{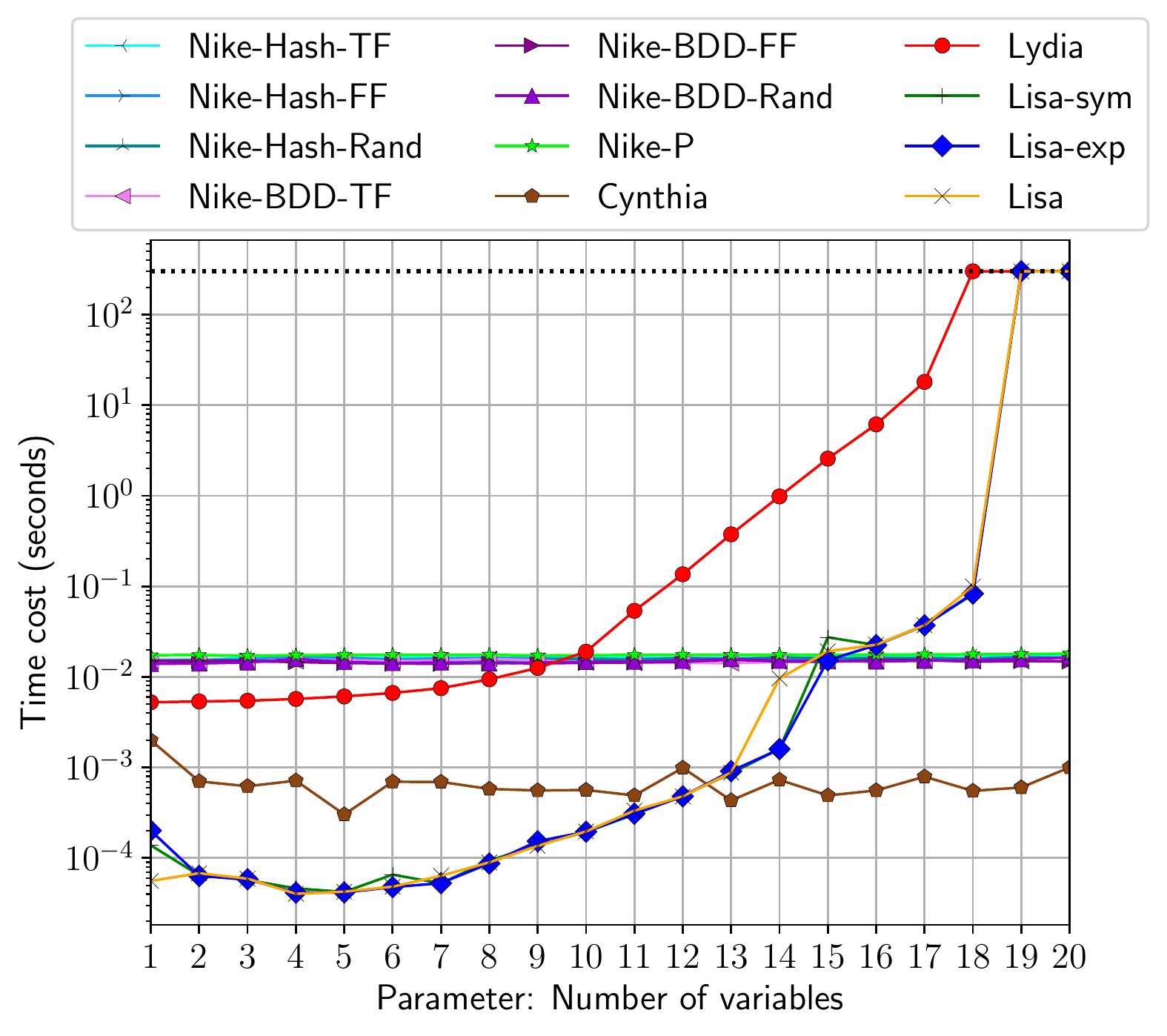}
\caption{$U$-pattern.} 
\label{fig:Uright}
\end{figure}

\begin{figure}[H]
\centering
\includegraphics[width=0.8\columnwidth]{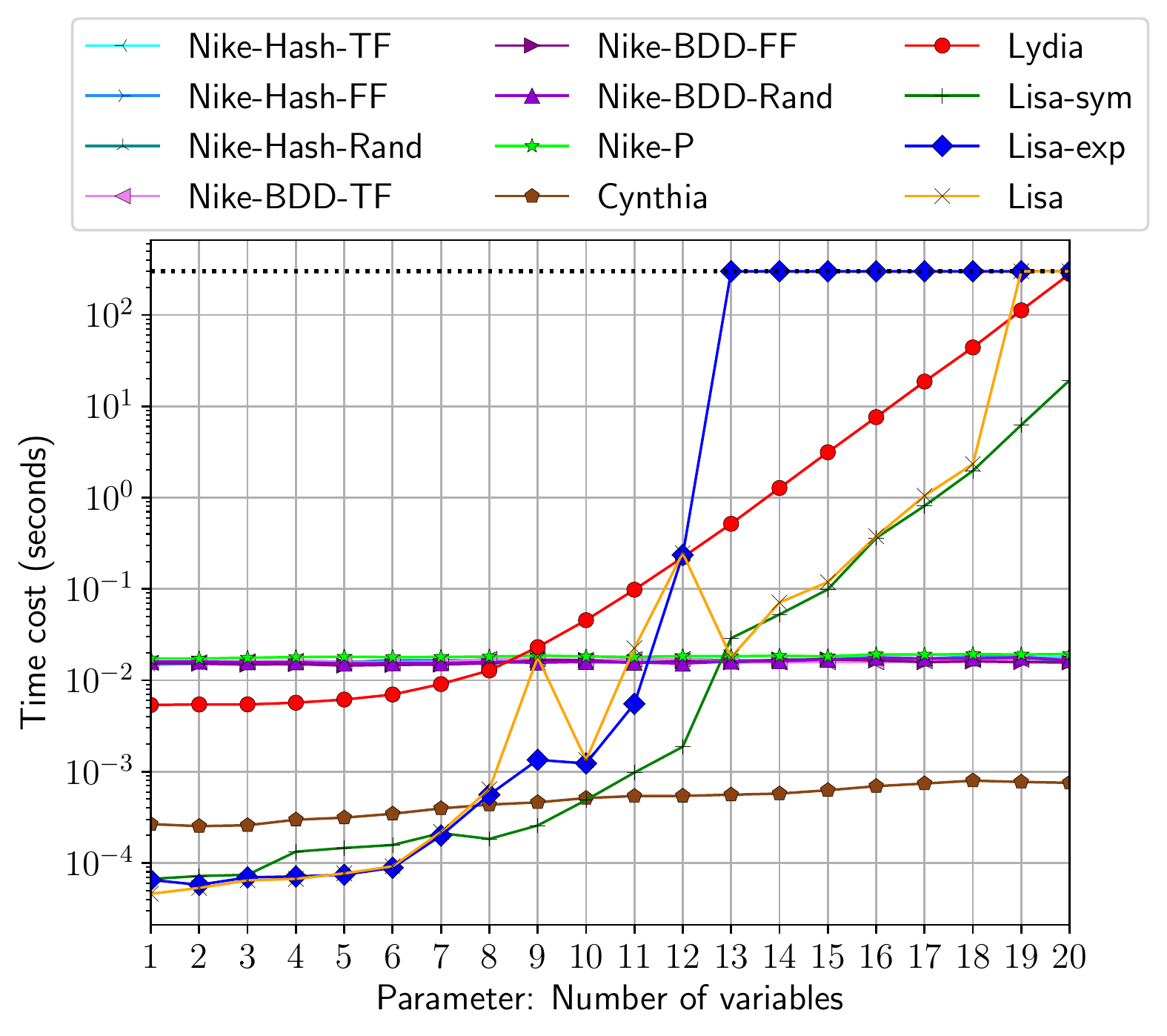}
\caption{$GF$-pattern.} 
\label{fig:Gfand}
\end{figure}

\begin{figure}[H]
\centering
\includegraphics[width=0.8\columnwidth]{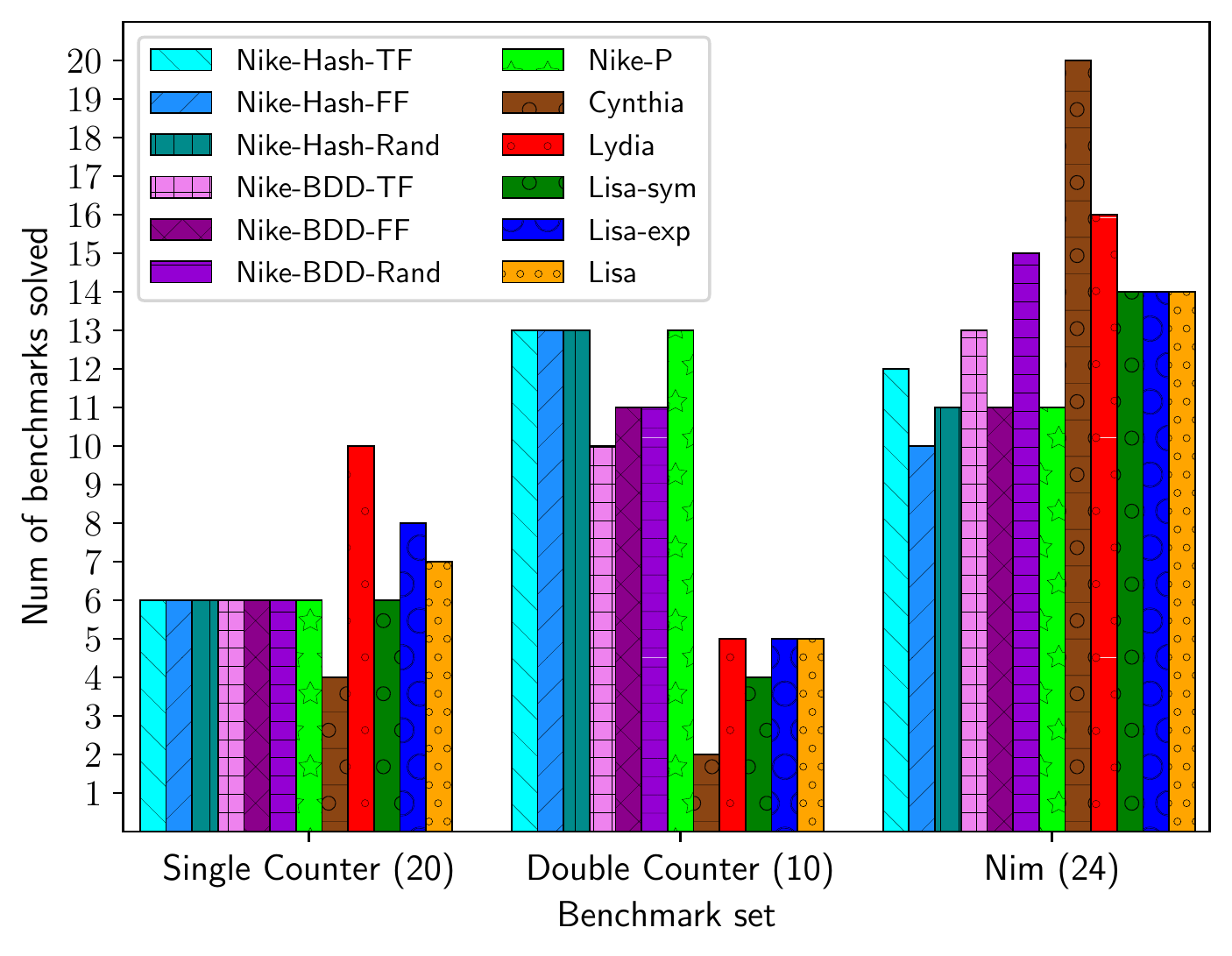}
\caption{\emph{Two-player-Games}.}
\label{fig:game}
\end{figure}

\begin{figure}[H]
\centering
\includegraphics[width=0.8\columnwidth]{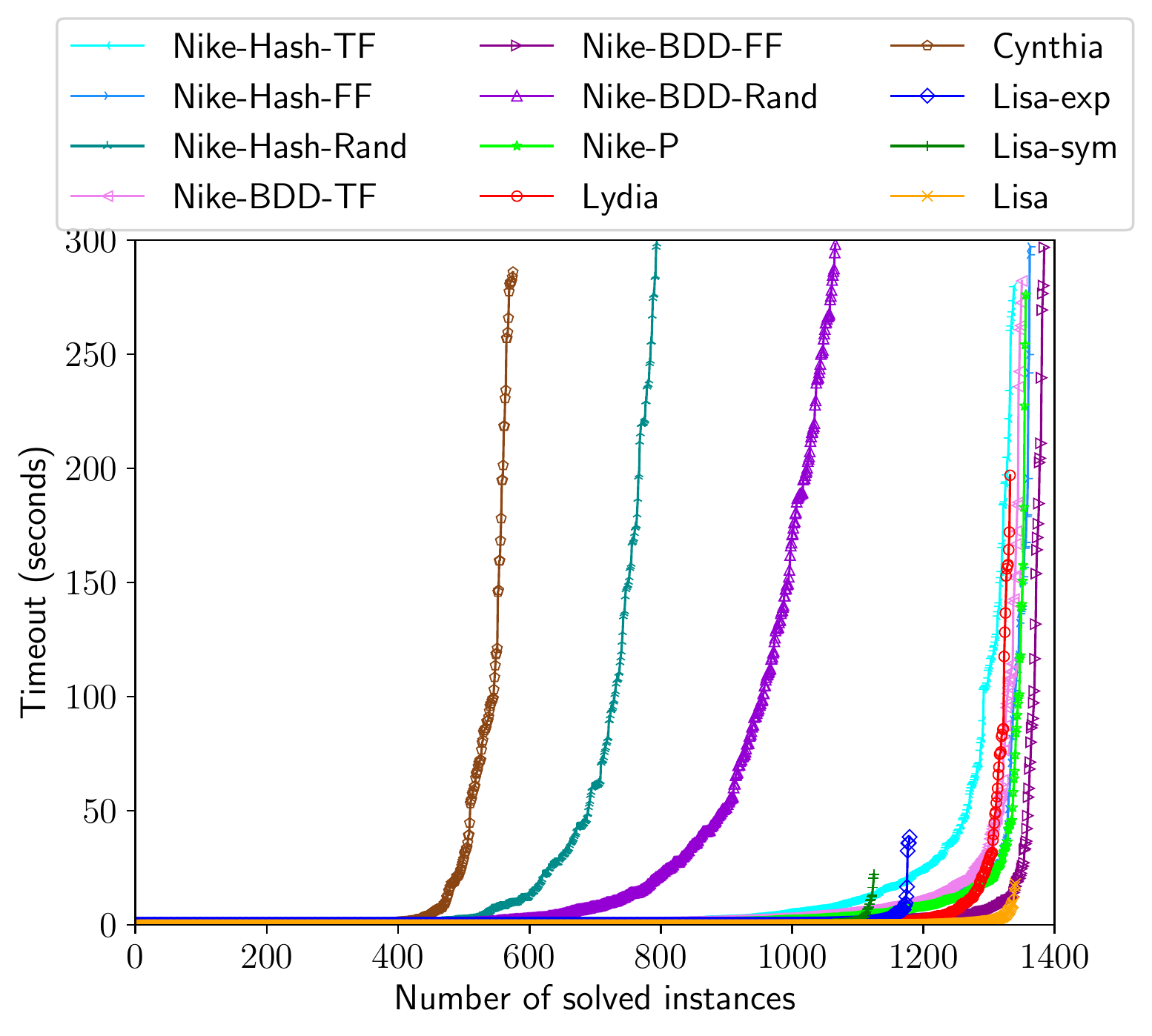}
\caption{\emph{Random}.}
\label{fig:random}
\end{figure}

\begin{figure}[H]
\centering
\includegraphics[width=0.8\columnwidth]{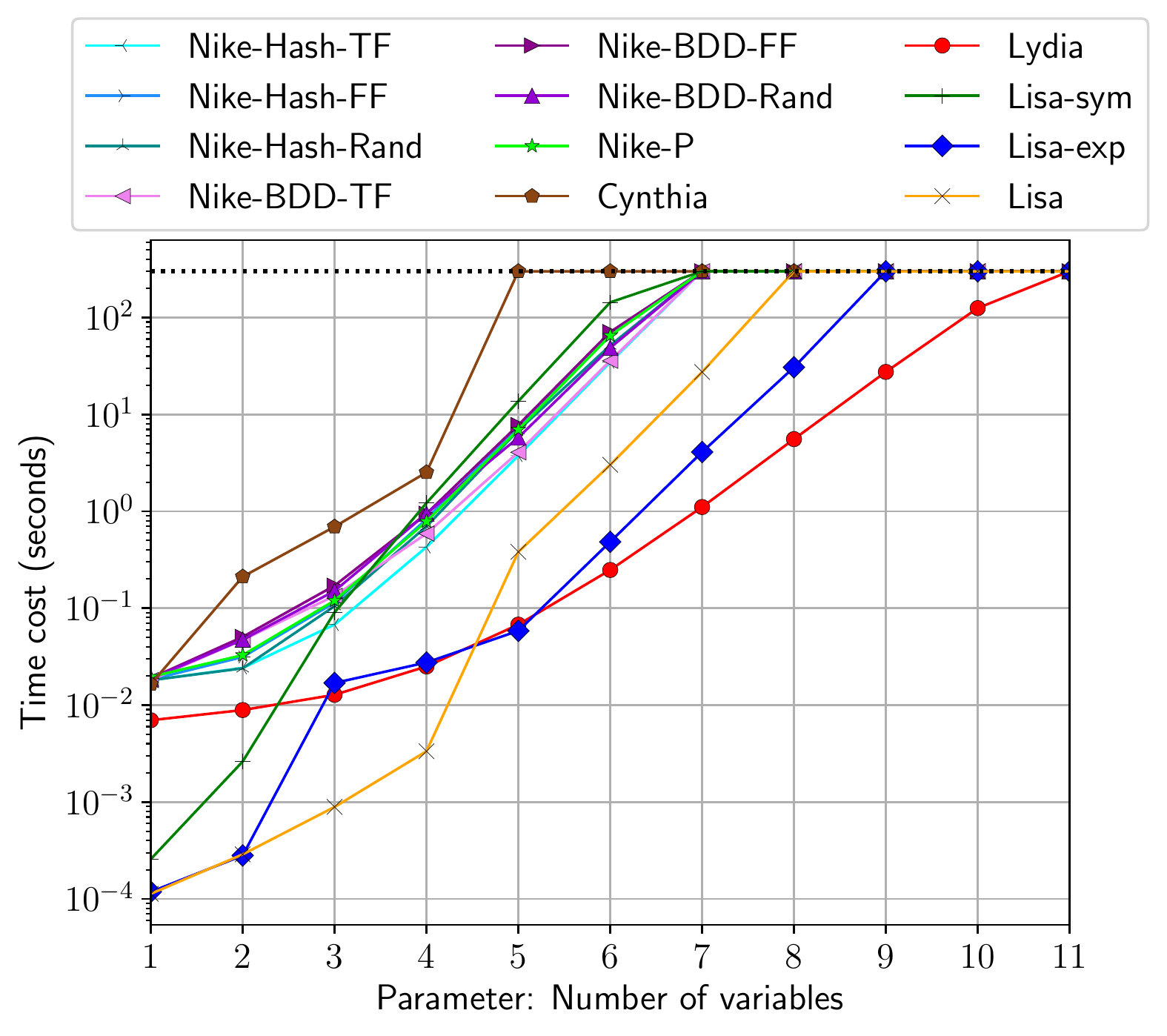}
\caption{Single Counter.} 
\label{fig:single-counter}
\end{figure}

\begin{figure}[H]
\centering
\includegraphics[width=0.8\columnwidth]{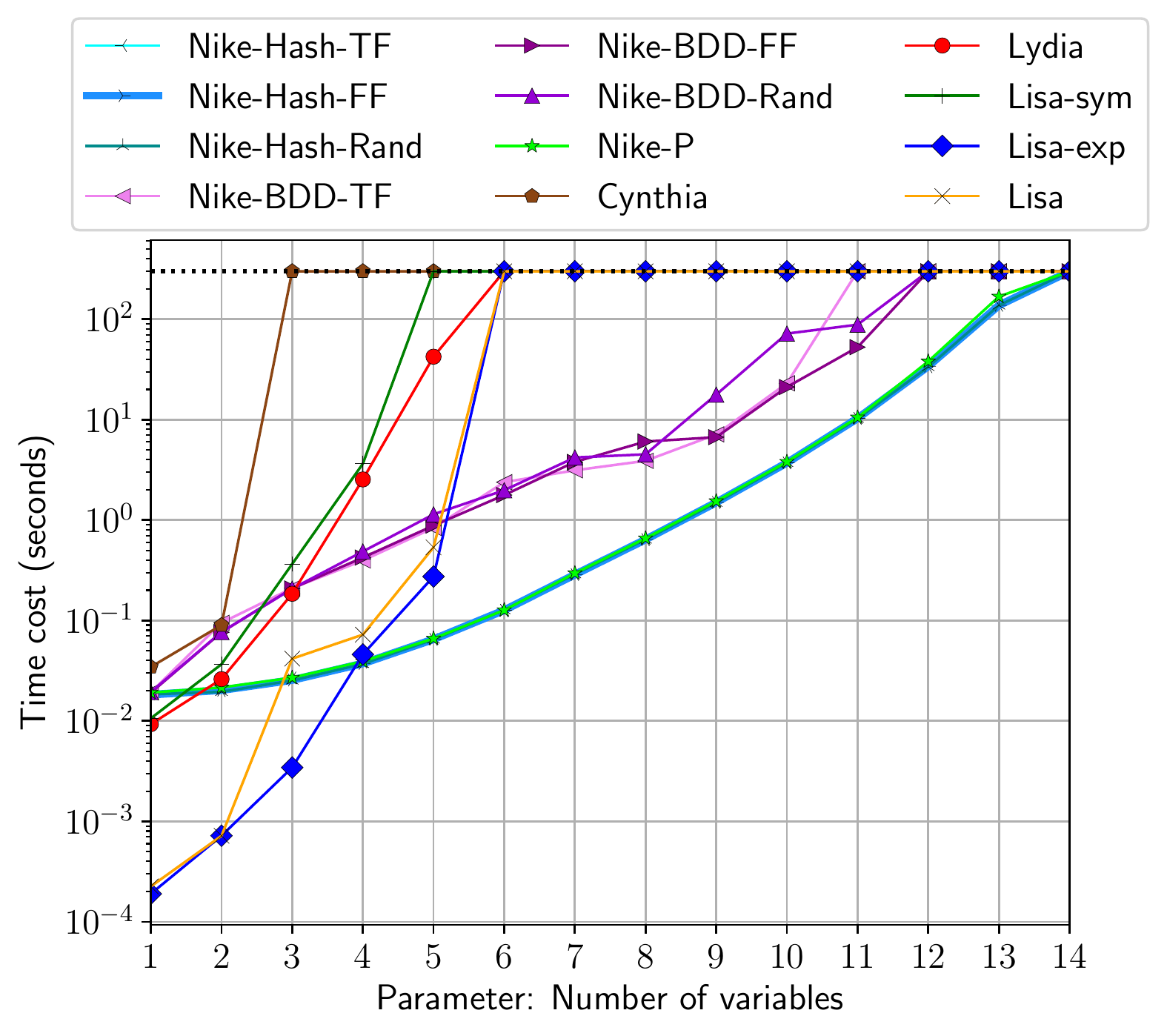}
\caption{Double Counter.} 
\label{fig:double-counter}
\end{figure}

\section{Results}
In this section, we show all the running times for each formula of each dataset.
\begin{table*}[h]
                \centering
                \caption{Results on \emph{Random Lydia Case 03 50}, time in milliseconds.}
                \resizebox{0.95\linewidth}{!}{ \renewcommand{\arraystretch}{1}
\begin{tabular}{|c|c|c|c|c|c|c|c|c|c|c|c|c|}
\hline
\textbf{name} & \textbf{Nike Hash True First} & \textbf{Nike Hash False First} & \textbf{Nike Hash Random} & \textbf{Nike Bdd True First} & \textbf{Nike Bdd False First} & \textbf{Nike Bdd Random} & \textbf{Nike Multithreaded} & \textbf{Cynthia} & \textbf{Lydia} & \textbf{Lisa Symbolic} & \textbf{Lisa Explicit} & \textbf{Lisa} \\ \hline
1 & 16.713829 & 17.150155 & 20.492264 & 17.892924 & 18.443531 & 19.430784 & 25.886465 & 14.689382 & 12.776132 & \textbf{0.06382} & 0.144459 & 0.151035 \\ \hline 
2 & 264.680185 & 114.455337 & 3557.0157 & 313.911392 & 84.343603 & 1209.118128 & 120.310362 & 257052.503168 & 359.320861 & \textbf{1.93298} & 1014.14 & 4.331 \\ \hline 
3 & 234.630113 & 128.765479 & 22201.139848 & 298.172512 & 96.743202 & 6647.968904 & 135.570123 & --- & 2123.071477 & \textbf{3.285} & 6682.4 & 9.94998 \\ \hline 
4 & 19.079327 & 19.830259 & 25.699179 & 18.866273 & 19.797592 & 25.574912 & 22.010723 & 13.133271 & 7.249208 & 0.040258 & \textbf{0.030786} & 0.031127 \\ \hline 
5 & 52.336247 & 19.384262 & 100.87368 & 49.863463 & 20.542235 & 84.771026 & 22.809315 & 689.867092 & 12.085266 & \textbf{0.140104} & 0.314111 & 0.275143 \\ \hline 
6 & 1577.386535 & 2093.895391 & --- & 725.614422 & 200.807375 & 1511.57062 & 2132.3953 & --- & 8.397469 & 0.057513 & 0.02679 & \textbf{0.022097} \\ \hline 
7 & 16.018741 & 16.307278 & 16.355825 & 15.962757 & 15.742274 & 16.11696 & 17.791634 & 0.334175 & 6.456763 & 0.061564 & 0.037286 & \textbf{0.034532} \\
\hline \end{tabular}
            }
            \label{tab:random-lydia-case-03-50}
            \end{table*}
\begin{table*}[h]
                \centering
                \caption{Results on \emph{Random Syft 04, 0-100}, time in milliseconds.}
                \resizebox{0.95\linewidth}{!}{ \renewcommand{\arraystretch}{1}
\begin{tabular}{|c|c|c|c|c|c|c|c|c|c|c|c|c|}
\hline
\textbf{name} & \textbf{Nike Hash True First} & \textbf{Nike Hash False First} & \textbf{Nike Hash Random} & \textbf{Nike Bdd True First} & \textbf{Nike Bdd False First} & \textbf{Nike Bdd Random} & \textbf{Nike Multithreaded} & \textbf{Cynthia} & \textbf{Lydia} & \textbf{Lisa Symbolic} & \textbf{Lisa Explicit} & \textbf{Lisa} \\ \hline
1 & 256.495355 & 228.735753 & 128248.419483 & 199.741663 & 100.188716 & 13191.074464 & 242.473378 & --- & 109.559054 & \textbf{1.8342} & 10.2948 & 9.92636 \\ \hline 
2 & 3161.904989 & 3587.597951 & --- & 1952.863484 & 1678.634163 & 149269.331504 & 3659.930257 & --- & 12.385282 & --- & 0.020014 & \textbf{0.020004} \\ \hline 
3 & 19.457204 & 19.707575 & 19.127515 & 19.056863 & 19.722886 & 18.932257 & 21.901975 & \textbf{1.387975} & 83.047175 & 3.64832 & 3.33221 & 3.26698 \\ \hline 
4 & 3902.566501 & 1866.096556 & 32343.601131 & 1948.857015 & 876.066626 & 32705.814231 & 1909.859299 & --- & 11.884163 & --- & 0.022698 & \textbf{0.021129} \\ \hline 
5 & 705.616097 & 532.830888 & --- & 673.623614 & 478.865033 & 25498.68394 & 547.91626 & --- & 4302.309411 & 57.2941 & 872.108 & \textbf{11.7207} \\ \hline 
6 & 11681.395936 & 3889.592045 & --- & 4164.122781 & 2355.545719 & 91223.570167 & 3948.020875 & --- & 15.180614 & 0.217421 & 0.101268 & \textbf{0.091401} \\ \hline 
7 & 43428.89856 & 1123.216954 & --- & 15303.335216 & 655.681661 & 187814.222779 & 1150.745076 & --- & 29.932109 & \textbf{0.855873} & 5.44999 & 5.41465 \\
\hline \end{tabular}
            }
            \label{tab:random-syft-04-0-100}
            \end{table*}
\begin{table*}[h]
                \centering
                \caption{Results on \emph{Random Syft 04, 100-200}, time in milliseconds.}
                \resizebox{0.95\linewidth}{!}{ \renewcommand{\arraystretch}{1}
\begin{tabular}{|c|c|c|c|c|c|c|c|c|c|c|c|c|}
\hline
\textbf{name} & \textbf{Nike Hash True First} & \textbf{Nike Hash False First} & \textbf{Nike Hash Random} & \textbf{Nike Bdd True First} & \textbf{Nike Bdd False First} & \textbf{Nike Bdd Random} & \textbf{Nike Multithreaded} & \textbf{Cynthia} & \textbf{Lydia} & \textbf{Lisa Symbolic} & \textbf{Lisa Explicit} & \textbf{Lisa} \\ \hline
101 & 295.002645 & 261.686271 & 30830.986106 & 242.378586 & 123.063882 & 2462.101056 & 266.229504 & --- & 12.177656 & \textbf{0.090917} & 0.111282 & 0.101919 \\ \hline 
102 & 1942.392116 & 3524.821693 & --- & 1339.39268 & 2501.494184 & 79674.19967 & 3554.085578 & --- & 40.152494 & \textbf{0.74228} & 1.00443 & 0.96298 \\ \hline 
103 & 25135.384243 & 978.008038 & --- & 8649.644533 & 436.282311 & 136539.597599 & 1000.435108 & --- & 12.680067 & \textbf{0.076829} & 0.086768 & 0.088021 \\ \hline 
104 & 23207.335679 & 2252.594011 & --- & 8433.401777 & 2147.376905 & 149511.378653 & 2280.106977 & --- & 37.502715 & \textbf{0.734245} & 1.03674 & 1.02668 \\ \hline 
105 & 63682.554619 & 3183.390487 & --- & 26685.771334 & 1057.304424 & --- & 3239.980873 & --- & 60.325145 & \textbf{1.85618} & 85.7931 & 81.7416 \\ \hline 
106 & 1528.021637 & 879.138803 & --- & 1445.727308 & 783.596915 & 53950.231467 & 888.384367 & --- & 22661.341436 & --- & --- & \textbf{241.9} \\ \hline 
107 & 18.541823 & 18.994679 & 19.118518 & 19.186762 & 19.22695 & 19.266799 & 21.059284 & 1.312632 & 19.04481 & \textbf{0.16584} & 0.357022 & 0.288947 \\
\hline \end{tabular}
            }
            \label{tab:random-syft-04-100-200}
            \end{table*}
\begin{table*}[h]
                \centering
                \caption{Results on \emph{Random Lydia Case 07 50}, time in milliseconds.}
                \resizebox{0.95\linewidth}{!}{ \renewcommand{\arraystretch}{1}
\begin{tabular}{|c|c|c|c|c|c|c|c|c|c|c|c|c|}
\hline
\textbf{name} & \textbf{Nike Hash True First} & \textbf{Nike Hash False First} & \textbf{Nike Hash Random} & \textbf{Nike Bdd True First} & \textbf{Nike Bdd False First} & \textbf{Nike Bdd Random} & \textbf{Nike Multithreaded} & \textbf{Cynthia} & \textbf{Lydia} & \textbf{Lisa Symbolic} & \textbf{Lisa Explicit} & \textbf{Lisa} \\ \hline
1 & --- & 139659.636414 & --- & 272613.224302 & 86346.68118 & --- & 150617.396927 & --- & --- & \textbf{6046.44} & --- & --- \\ \hline 
2 & 108246.403146 & 115.779725 & --- & 46391.418416 & 191.093623 & --- & 118.8012 & --- & 1876.379831 & \textbf{27.1607} & 12457.7 & 80.1553 \\ \hline 
3 & --- & --- & --- & --- & --- & --- & --- & --- & --- & --- & --- & --- \\ \hline 
4 & 17.566906 & 17.394854 & 17.141297 & 17.420313 & 17.46057 & 16.627154 & 18.922813 & 0.912246 & 17.759666 & \textbf{0.055737} & 0.143472 & 0.120822 \\ \hline 
5 & 136.994026 & 23.104773 & 914.861403 & 172.446904 & 24.219817 & 1016.665293 & 25.182862 & 21372.884565 & 130.019039 & \textbf{2.96725} & 7.03019 & 6.27644 \\ \hline 
6 & 17.232201 & 16.426434 & 17.360178 & 16.669022 & 15.882919 & 16.582951 & 18.614869 & 0.845496 & 10.277482 & \textbf{0.034572} & 0.057241 & 0.04242 \\ \hline 
7 & 32592.043851 & 37890.665585 & --- & 30436.587611 & \textbf{21635.183297} & --- & 45652.719643 & --- & --- & --- & --- & --- \\
\hline \end{tabular}
            }
            \label{tab:random-lydia-case-07-50}
            \end{table*}
\begin{table*}[h]
                \centering
                \caption{Results on \emph{Random Syft 05, 100-200}, time in milliseconds.}
                \resizebox{0.95\linewidth}{!}{ \renewcommand{\arraystretch}{1}
\begin{tabular}{|c|c|c|c|c|c|c|c|c|c|c|c|c|}
\hline
\textbf{name} & \textbf{Nike Hash True First} & \textbf{Nike Hash False First} & \textbf{Nike Hash Random} & \textbf{Nike Bdd True First} & \textbf{Nike Bdd False First} & \textbf{Nike Bdd Random} & \textbf{Nike Multithreaded} & \textbf{Cynthia} & \textbf{Lydia} & \textbf{Lisa Symbolic} & \textbf{Lisa Explicit} & \textbf{Lisa} \\ \hline
101 & 11698.55439 & 3417.715229 & --- & 5215.458081 & 680.807943 & --- & 3413.549859 & --- & 39.664408 & \textbf{24.8839} & 48.0834 & 46.7791 \\ \hline 
102 & 38812.320531 & 3398.762652 & --- & 19117.586499 & \textbf{1162.067416} & --- & 3408.450807 & --- & 30244.842227 & --- & --- & 1947.48 \\ \hline 
103 & 89455.335706 & 27094.714908 & --- & 43400.077024 & 22596.117104 & --- & 43730.816842 & --- & 842.856736 & 6.19327 & 218.777 & \textbf{3.72073} \\ \hline 
104 & 11103.102123 & 5989.265511 & --- & 4251.804433 & 2166.598487 & 187442.459665 & 6022.333877 & --- & 11.986854 & --- & 0.021689 & \textbf{0.02161} \\ \hline 
105 & 20286.139818 & 1345.61415 & --- & 8148.018682 & 403.962658 & 98424.359682 & 1352.164198 & --- & 12.967107 & --- & 0.021439 & \textbf{0.021326} \\ \hline 
106 & 16699.783977 & 3311.536535 & --- & 7725.362959 & 1126.600651 & --- & 3305.053633 & --- & 12.411959 & --- & 0.023527 & \textbf{0.023249} \\ \hline 
107 & 213347.76552 & 138242.852168 & --- & 76715.798005 & 18652.59465 & --- & 140503.700835 & --- & 57.010477 & 11.235 & 2.03394 & \textbf{1.87405} \\
\hline \end{tabular}
            }
            \label{tab:random-syft-05-100-200}
            \end{table*}
\begin{table*}[h]
                \centering
                \caption{Results on \emph{Random Syft 05, 0-100}, time in milliseconds.}
                \resizebox{0.95\linewidth}{!}{ \renewcommand{\arraystretch}{1}
\begin{tabular}{|c|c|c|c|c|c|c|c|c|c|c|c|c|}
\hline
\textbf{name} & \textbf{Nike Hash True First} & \textbf{Nike Hash False First} & \textbf{Nike Hash Random} & \textbf{Nike Bdd True First} & \textbf{Nike Bdd False First} & \textbf{Nike Bdd Random} & \textbf{Nike Multithreaded} & \textbf{Cynthia} & \textbf{Lydia} & \textbf{Lisa Symbolic} & \textbf{Lisa Explicit} & \textbf{Lisa} \\ \hline
1 & 58700.330788 & 2990.400872 & --- & 20294.760652 & 2017.257092 & --- & 3045.018323 & --- & 73.348249 & \textbf{4.43764} & 36.4456 & 33.1644 \\ \hline 
2 & 141128.69629 & 137284.051836 & --- & 28655.121064 & 27220.550187 & --- & 139290.802305 & --- & 17.210375 & 0.169509 & 0.080361 & \textbf{0.069687} \\ \hline 
3 & 8940.169389 & 3308.628316 & --- & 4689.274945 & 2791.597485 & 188405.001144 & 3311.865263 & --- & 35.006535 & \textbf{0.731498} & 2.36933 & 2.20971 \\ \hline 
4 & 7816.447644 & 767.415549 & --- & 5432.115451 & 590.001574 & 284741.627653 & 771.415557 & --- & 1414.800526 & 858.889 & --- & \textbf{220.26} \\ \hline 
5 & 2482.345045 & 169.243278 & --- & 2411.419836 & \textbf{167.680277} & 74869.149956 & 172.670223 & --- & 28740.409431 & 693.476 & --- & 1802.5 \\ \hline 
6 & 1229.462051 & 824.254725 & --- & 1252.748591 & 594.600642 & --- & 828.221929 & --- & 10742.37776 & \textbf{45.1448} & --- & 53.5212 \\ \hline 
7 & 422.816407 & 259.31367 & 79747.567677 & 390.160987 & 251.17055 & 34541.595055 & 260.90415 & --- & 9.664667 & --- & 0.061738 & \textbf{0.019271} \\
\hline \end{tabular}
            }
            \label{tab:random-syft-05-0-100}
            \end{table*}
\begin{table*}[h]
                \centering
                \caption{Results on \emph{Random Syft 01, 0-100}, time in milliseconds.}
                \resizebox{0.95\linewidth}{!}{ \renewcommand{\arraystretch}{1}
\begin{tabular}{|c|c|c|c|c|c|c|c|c|c|c|c|c|}
\hline
\textbf{name} & \textbf{Nike Hash True First} & \textbf{Nike Hash False First} & \textbf{Nike Hash Random} & \textbf{Nike Bdd True First} & \textbf{Nike Bdd False First} & \textbf{Nike Bdd Random} & \textbf{Nike Multithreaded} & \textbf{Cynthia} & \textbf{Lydia} & \textbf{Lisa Symbolic} & \textbf{Lisa Explicit} & \textbf{Lisa} \\ \hline
1 & 17.838381 & 18.010678 & 19.426528 & 18.434495 & 18.272431 & 19.153742 & 25.453386 & 1.480109 & 17.149557 & 0.06058 & 0.097848 & \textbf{0.041566} \\ \hline 
2 & 15.555838 & 16.936421 & 16.358218 & 16.154483 & 15.942102 & 17.039143 & 18.64216 & 0.433352 & 17.302952 & 0.21697 & 0.087593 & \textbf{0.08387} \\ \hline 
3 & 244.892236 & 252.21823 & 2774.826025 & 264.237455 & 273.860434 & 2827.775795 & 256.829823 & 7029.127958 & 10.39073 & --- & 0.022478 & \textbf{0.018531} \\ \hline 
4 & 17.890429 & 18.123882 & 18.386884 & 18.493321 & 18.205205 & 18.900738 & 19.865162 & 0.95174 & 11.366788 & 0.043558 & 0.055002 & \textbf{0.040489} \\ \hline 
5 & 2304.336286 & 2557.382003 & --- & 2695.759816 & 554.392568 & 267737.691448 & 2574.013279 & --- & 94.161289 & 205.777 & 3.52906 & \textbf{2.99827} \\ \hline 
6 & 19.385058 & 19.853257 & 18.833662 & 19.287574 & 19.506592 & 19.331106 & 21.038914 & 1.37929 & 14.260233 & \textbf{0.048421} & 0.048696 & 0.048428 \\ \hline 
7 & 16.689495 & 17.907299 & 18.143165 & 18.205655 & 17.339884 & 17.51406 & 19.608336 & \textbf{0.808236} & 73.956006 & 35.799 & 5.08427 & 4.73567 \\
\hline \end{tabular}
            }
            \label{tab:random-syft-01-0-100}
            \end{table*}
\begin{table*}[h]
                \centering
                \caption{Results on \emph{Random Syft 01, 100-200}, time in milliseconds.}
                \resizebox{0.95\linewidth}{!}{ \renewcommand{\arraystretch}{1}
\begin{tabular}{|c|c|c|c|c|c|c|c|c|c|c|c|c|}
\hline
\textbf{name} & \textbf{Nike Hash True First} & \textbf{Nike Hash False First} & \textbf{Nike Hash Random} & \textbf{Nike Bdd True First} & \textbf{Nike Bdd False First} & \textbf{Nike Bdd Random} & \textbf{Nike Multithreaded} & \textbf{Cynthia} & \textbf{Lydia} & \textbf{Lisa Symbolic} & \textbf{Lisa Explicit} & \textbf{Lisa} \\ \hline
101 & 17.049827 & 17.433849 & 18.03419 & 18.127264 & 18.324177 & 18.168057 & 19.728269 & 0.986675 & 10.005363 & 0.036104 & \textbf{0.035689} & 0.037083 \\ \hline 
102 & 17.298718 & 18.766137 & 17.739139 & 18.916295 & 19.015619 & 18.320205 & 20.296427 & 1.240988 & 14.055802 & 0.042816 & 0.043886 & \textbf{0.040517} \\ \hline 
103 & 2619.760211 & 165036.284525 & --- & 3028.209367 & 1865.938059 & --- & 5049.075983 & --- & 83.840387 & 6.97416 & 5.87934 & \textbf{5.43837} \\ \hline 
104 & 73.593391 & 79.033897 & 2107.287044 & 85.097798 & 50.440105 & 51.304254 & 79.462928 & --- & 16.982128 & 0.517877 & 0.09886 & \textbf{0.08593} \\ \hline 
105 & 18.612943 & 19.336494 & 18.767024 & 19.363399 & 19.432526 & 19.680179 & 20.691308 & 1.296773 & 14.172291 & 0.046593 & \textbf{0.041691} & 0.042389 \\ \hline 
106 & 65.523752 & 68.227214 & 1445.416127 & 75.811361 & 47.489396 & 432.719049 & 70.144778 & --- & 18.532368 & 0.284118 & 0.118855 & \textbf{0.091792} \\ \hline 
107 & 2312.880279 & 2561.363222 & --- & 2707.34949 & 561.11436 & 267230.115431 & 2567.199951 & --- & 91.060229 & 204.388 & 3.63981 & \textbf{2.78289} \\
\hline \end{tabular}
            }
            \label{tab:random-syft-01-100-200}
            \end{table*}
\begin{table*}[h]
                \centering
                \caption{Results on \emph{Random Syft 02, 0-100}, time in milliseconds.}
                \resizebox{0.95\linewidth}{!}{ \renewcommand{\arraystretch}{1}
\begin{tabular}{|c|c|c|c|c|c|c|c|c|c|c|c|c|}
\hline
\textbf{name} & \textbf{Nike Hash True First} & \textbf{Nike Hash False First} & \textbf{Nike Hash Random} & \textbf{Nike Bdd True First} & \textbf{Nike Bdd False First} & \textbf{Nike Bdd Random} & \textbf{Nike Multithreaded} & \textbf{Cynthia} & \textbf{Lydia} & \textbf{Lisa Symbolic} & \textbf{Lisa Explicit} & \textbf{Lisa} \\ \hline
1 & 1504.973108 & 5157.85633 & --- & 1857.615553 & 2446.271323 & 6742.374297 & 2886.465606 & --- & 4590.170819 & 21.8974 & 978.013 & \textbf{7.64967} \\ \hline 
2 & 18.706711 & 18.533534 & 19.349636 & 18.803158 & 19.292196 & 19.238666 & 20.782447 & 1.245319 & 14.055715 & 0.022093 & 0.023939 & \textbf{0.021903} \\ \hline 
3 & 3552.730691 & 290.086476 & --- & 999.373529 & 276.591712 & 15301.785696 & 291.288737 & 66763.036325 & 14.90648 & 0.076225 & 0.078319 & \textbf{0.062949} \\ \hline 
4 & 4775.026745 & 977.8761 & --- & 3626.136501 & 249.937242 & 250056.665951 & 987.554438 & --- & 163.610671 & 93.8458 & 13.697 & \textbf{11.1197} \\ \hline 
5 & 2450.454876 & 2829.232249 & --- & 2654.521757 & 1876.844986 & --- & 2835.214667 & --- & \textbf{321.614877} & --- & 397.263 & 347.974 \\ \hline 
6 & 9297.348399 & 1830.678317 & --- & 2201.321798 & 1737.504453 & 24281.761385 & 1844.30975 & --- & 20.38892 & 0.039756 & 0.034423 & \textbf{0.034081} \\ \hline 
7 & 18.558617 & 19.291644 & 19.592553 & 19.768541 & 19.282277 & 19.088439 & 21.151468 & 1.419 & 12.99837 & 0.035211 & 0.032893 & \textbf{0.028164} \\
\hline \end{tabular}
            }
            \label{tab:random-syft-02-0-100}
            \end{table*}
\begin{table*}[h]
                \centering
                \caption{Results on \emph{Random Syft 02, 100-200}, time in milliseconds.}
                \resizebox{0.95\linewidth}{!}{ \renewcommand{\arraystretch}{1}
\begin{tabular}{|c|c|c|c|c|c|c|c|c|c|c|c|c|}
\hline
\textbf{name} & \textbf{Nike Hash True First} & \textbf{Nike Hash False First} & \textbf{Nike Hash Random} & \textbf{Nike Bdd True First} & \textbf{Nike Bdd False First} & \textbf{Nike Bdd Random} & \textbf{Nike Multithreaded} & \textbf{Cynthia} & \textbf{Lydia} & \textbf{Lisa Symbolic} & \textbf{Lisa Explicit} & \textbf{Lisa} \\ \hline
101 & 1893.295514 & 2124.460514 & --- & 1938.579287 & 1702.496663 & --- & 2154.451071 & --- & 434.622776 & --- & 655.528 & \textbf{107.261} \\ \hline 
102 & 2319.958785 & 846.097249 & --- & 844.271099 & 518.399692 & 22340.772101 & 858.199308 & 145879.742948 & 12.540078 & --- & \textbf{0.022875} & 0.024112 \\ \hline 
103 & 144.770494 & 165.784182 & 7679.685127 & 158.743126 & 116.606741 & 361.039836 & 170.732243 & --- & 97.028037 & \textbf{1.5045} & 2.49803 & 2.38788 \\ \hline 
104 & 18.659551 & 18.373764 & 18.109879 & 18.437595 & 18.390078 & 18.653903 & 20.269547 & 1.430386 & 10.908449 & 0.083266 & \textbf{0.045452} & 0.065573 \\ \hline 
105 & 133.218373 & 165.405762 & 12514.192894 & 97.025729 & 69.989218 & 605.606859 & 174.44256 & --- & 10.511282 & 0.078108 & 0.05796 & \textbf{0.056922} \\ \hline 
106 & 62.570635 & 59.888824 & 790.981212 & 62.905215 & 61.995424 & 586.848967 & 63.477016 & 1159.876483 & 161.391046 & 30.2802 & 10.8114 & \textbf{10.4937} \\ \hline 
107 & 603.904043 & 641.253886 & 110485.987951 & 621.13678 & 621.855499 & 73897.325208 & 653.651782 & --- & 330.434515 & 352.153 & 261.918 & \textbf{243.038} \\
\hline \end{tabular}
            }
            \label{tab:random-syft-02-100-200}
            \end{table*}
\begin{table*}[h]
                \centering
                \caption{Results on \emph{Uright}, time in milliseconds.}
                \resizebox{0.95\linewidth}{!}{ \renewcommand{\arraystretch}{1}
\begin{tabular}{|c|c|c|c|c|c|c|c|c|c|c|c|c|}
\hline
\textbf{name} & \textbf{Nike Hash True First} & \textbf{Nike Hash False First} & \textbf{Nike Hash Random} & \textbf{Nike Bdd True First} & \textbf{Nike Bdd False First} & \textbf{Nike Bdd Random} & \textbf{Nike Multithreaded} & \textbf{Cynthia} & \textbf{Lydia} & \textbf{Lisa Symbolic} & \textbf{Lisa Explicit} & \textbf{Lisa} \\ \hline
uright01 & 14.999567 & 15.462853 & 14.883559 & 14.865238 & 14.998477 & 13.940153 & 17.356649 & 1.995103 & 5.247303 & 0.138404 & 0.200602 & \textbf{0.055955} \\ \hline 
uright02 & 15.094761 & 15.38919 & 15.500275 & 14.319042 & 14.88239 & 14.062963 & 17.449376 & 0.702015 & 5.354601 & 0.064368 & \textbf{0.063398} & 0.068017 \\ \hline 
uright03 & 16.635718 & 15.695066 & 15.027712 & 14.483598 & 15.092849 & 14.438113 & 17.078239 & 0.620635 & 5.450475 & \textbf{0.057319} & 0.058338 & 0.059256 \\ \hline 
uright04 & 16.45286 & 16.809773 & 15.015942 & 14.84263 & 14.561528 & 15.816918 & 17.335316 & 0.715294 & 5.706349 & 0.046205 & 0.041596 & \textbf{0.040328} \\ \hline 
uright05 & 16.294455 & 16.307927 & 15.276146 & 15.317786 & 14.232395 & 14.600184 & 17.449579 & 0.30296 & 6.080545 & 0.042322 & \textbf{0.041693} & 0.042551 \\ \hline 
uright06 & 16.366488 & 15.497276 & 15.429629 & 15.02954 & 14.025922 & 14.222141 & 17.509895 & 0.695206 & 6.644925 & 0.065597 & \textbf{0.047732} & 0.048302 \\ \hline 
uright07 & 16.42893 & 16.002619 & 15.165841 & 15.164762 & 13.981411 & 14.432798 & 17.490425 & 0.692417 & 7.502385 & \textbf{0.052625} & 0.052897 & 0.063675 \\
\hline \end{tabular}
            }
            \label{tab:Uright}
            \end{table*}
\begin{table*}[h]
                \centering
                \caption{Results on \emph{Gfand}, time in milliseconds.}
                \resizebox{0.95\linewidth}{!}{ \renewcommand{\arraystretch}{1}
\begin{tabular}{|c|c|c|c|c|c|c|c|c|c|c|c|c|}
\hline
\textbf{name} & \textbf{Nike Hash True First} & \textbf{Nike Hash False First} & \textbf{Nike Hash Random} & \textbf{Nike Bdd True First} & \textbf{Nike Bdd False First} & \textbf{Nike Bdd Random} & \textbf{Nike Multithreaded} & \textbf{Cynthia} & \textbf{Lydia} & \textbf{Lisa Symbolic} & \textbf{Lisa Explicit} & \textbf{Lisa} \\ \hline
gfand01 & 14.664404 & 15.476503 & 16.250197 & 15.496387 & 15.251414 & 16.007296 & 17.336509 & 0.266138 & 5.364753 & 0.066774 & 0.065305 & \textbf{0.045994} \\ \hline 
gfand02 & 15.081269 & 15.645543 & 15.272277 & 15.701056 & 15.220697 & 16.059728 & 17.27196 & 0.253074 & 5.435404 & 0.071976 & 0.057631 & \textbf{0.053326} \\ \hline 
gfand03 & 15.476061 & 15.80807 & 15.131336 & 16.069496 & 14.850944 & 15.798273 & 17.655924 & 0.258545 & 5.442908 & 0.074251 & 0.069557 & \textbf{0.064516} \\ \hline 
gfand04 & 15.985911 & 15.636417 & 16.25583 & 16.285712 & 15.024169 & 15.766807 & 17.971768 & 0.297777 & 5.679279 & 0.132932 & 0.071209 & \textbf{0.066989} \\ \hline 
gfand05 & 16.089714 & 15.721872 & 16.024327 & 15.85085 & 14.484887 & 15.269209 & 18.114179 & 0.313383 & 6.147079 & 0.146031 & \textbf{0.074355} & 0.076898 \\ \hline 
gfand06 & 15.740626 & 16.710326 & 15.815687 & 15.737012 & 14.742387 & 15.401362 & 17.966487 & 0.34558 & 6.979331 & 0.157342 & \textbf{0.088702} & 0.091956 \\ \hline 
gfand07 & 15.730463 & 16.658157 & 15.415758 & 15.91973 & 14.870688 & 15.4181 & 18.026563 & 0.395522 & 9.083914 & 0.212142 & \textbf{0.199394} & 0.2188 \\
\hline \end{tabular}
            }
            \label{tab:GFand}
            \end{table*}
\begin{table*}[h]
                \centering
                \caption{Results on \emph{Nim 04}, time in milliseconds.}
                \resizebox{0.95\linewidth}{!}{ \renewcommand{\arraystretch}{1}
\begin{tabular}{|c|c|c|c|c|c|c|c|c|c|c|c|c|}
\hline
\textbf{name} & \textbf{Nike Hash True First} & \textbf{Nike Hash False First} & \textbf{Nike Hash Random} & \textbf{Nike Bdd True First} & \textbf{Nike Bdd False First} & \textbf{Nike Bdd Random} & \textbf{Nike Multithreaded} & \textbf{Cynthia} & \textbf{Lydia} & \textbf{Lisa Symbolic} & \textbf{Lisa Explicit} & \textbf{Lisa} \\ \hline
nim\_04\_01 & 200248.683865 & 229673.2003 & 251716.519243 & 80051.226611 & 99057.523944 & 83007.538629 & 261215.876355 & 8614.624562 & 19574.724446 & 6523.22 & \textbf{6.79701} & 22.7605 \\ \hline 
nim\_04\_02 & --- & --- & --- & --- & --- & --- & --- & --- & --- & --- & --- & --- \\ \hline 
nim\_04\_03 & --- & --- & --- & --- & --- & --- & --- & --- & --- & --- & --- & --- \\ \hline 
nim\_04\_04 & --- & --- & --- & --- & --- & --- & --- & --- & --- & --- & --- & --- \\ \hline 
nim\_04\_05 & --- & --- & --- & --- & --- & --- & --- & --- & --- & --- & --- & --- \\ \hline 
nim\_04\_06 & --- & --- & --- & --- & --- & --- & --- & --- & --- & --- & --- & --- \\ \hline 
nim\_04\_07 & --- & --- & --- & --- & --- & --- & --- & --- & --- & --- & --- & --- \\ \hline 
nim\_04\_08 & --- & --- & --- & --- & --- & --- & --- & --- & --- & --- & --- & --- \\ \hline 
nim\_04\_09 & --- & --- & --- & --- & --- & --- & --- & --- & --- & --- & --- & --- \\ \hline 
nim\_04\_10 & --- & --- & --- & --- & --- & --- & --- & --- & --- & --- & --- & --- \\ \hline \end{tabular}
            }
            \label{tab:nim-04}
            \end{table*}
\begin{table*}[h]
                \centering
                \caption{Results on \emph{Random Lydia Case 05 50}, time in milliseconds.}
                \resizebox{0.95\linewidth}{!}{ \renewcommand{\arraystretch}{1}
\begin{tabular}{|c|c|c|c|c|c|c|c|c|c|c|c|c|}
\hline
\textbf{name} & \textbf{Nike Hash True First} & \textbf{Nike Hash False First} & \textbf{Nike Hash Random} & \textbf{Nike Bdd True First} & \textbf{Nike Bdd False First} & \textbf{Nike Bdd Random} & \textbf{Nike Multithreaded} & \textbf{Cynthia} & \textbf{Lydia} & \textbf{Lisa Symbolic} & \textbf{Lisa Explicit} & \textbf{Lisa} \\ \hline
1 & 17186.954282 & 2956.285242 & --- & 5894.234357 & 1710.078046 & --- & 2955.30058 & --- & 21309.590105 & 78.1458 & --- & \textbf{61.445} \\ \hline 
2 & 5590.052495 & 5562.885684 & --- & 5285.853787 & 3710.628698 & --- & 5620.573055 & --- & 164366.265304 & \textbf{40.5712} & --- & 2756.39 \\ \hline 
3 & 19.681307 & 19.744672 & 19.845728 & 19.39953 & 18.779845 & 19.490557 & 21.201815 & 1.897343 & 11.232309 & 0.041901 & 0.072184 & \textbf{0.028203} \\ \hline 
4 & 1302.535106 & 20.075408 & --- & 909.670963 & 19.545558 & 656.361316 & 21.515379 & 4480.153605 & 93.623848 & \textbf{0.263688} & 2.7205 & 1.55379 \\ \hline 
5 & 5573.333707 & 3757.327975 & --- & 5286.536226 & 2344.961407 & --- & 3817.344839 & --- & --- & \textbf{1439.49} & --- & 3940.64 \\ \hline 
6 & --- & 22450.212368 & --- & 107409.098376 & 10054.463637 & --- & 25579.55604 & --- & 65789.442762 & \textbf{147.514} & --- & 209.235 \\ \hline 
7 & 16.814876 & 16.623774 & 16.986842 & 16.689582 & 16.455677 & 17.189713 & 18.338521 & 2.025441 & 9.801064 & 0.068553 & 0.077106 & \textbf{0.036585} \\
\hline \end{tabular}
            }
            \label{tab:random-lydia-case-05-50}
            \end{table*}
\begin{table*}[h]
                \centering
                \caption{Results on \emph{Random Lydia Case 06 50}, time in milliseconds.}
                \resizebox{0.95\linewidth}{!}{ \renewcommand{\arraystretch}{1}
\begin{tabular}{|c|c|c|c|c|c|c|c|c|c|c|c|c|}
\hline
\textbf{name} & \textbf{Nike Hash True First} & \textbf{Nike Hash False First} & \textbf{Nike Hash Random} & \textbf{Nike Bdd True First} & \textbf{Nike Bdd False First} & \textbf{Nike Bdd Random} & \textbf{Nike Multithreaded} & \textbf{Cynthia} & \textbf{Lydia} & \textbf{Lisa Symbolic} & \textbf{Lisa Explicit} & \textbf{Lisa} \\ \hline
1 & 17.725612 & 17.815749 & 18.306462 & 17.260368 & 16.949622 & 18.316431 & 26.258953 & \textbf{1.824316} & 74788.499671 & 128.161 & --- & 1605.91 \\ \hline 
2 & 17.622506 & 18.725651 & 28.870486 & 20.297471 & 19.637225 & 29.361514 & 21.549324 & 88.152674 & 108.293412 & \textbf{0.449279} & 1.11059 & 0.61805 \\ \hline 
3 & 2647.044011 & 39.104455 & 147931.751387 & 1283.700043 & 50.225297 & 15786.606906 & 42.017331 & --- & 112.983522 & \textbf{1.57557} & 48.2856 & 46.751 \\ \hline 
4 & 16.78076 & 16.583245 & 16.804801 & 16.658298 & 15.80621 & 16.467267 & 18.624064 & 2.675993 & 10.65112 & \textbf{0.043024} & 0.045095 & 0.047252 \\ \hline 
5 & --- & --- & --- & --- & --- & --- & --- & --- & 75795.185641 & \textbf{198.599} & --- & 312.79 \\ \hline 
6 & 12320.366361 & 64.441768 & --- & 5132.779359 & 42.093978 & 3864.07905 & 67.630661 & 91133.936828 & 25.545632 & 0.353424 & 0.384728 & \textbf{0.318895} \\ \hline 
7 & --- & --- & --- & --- & 32863.42443 & --- & --- & --- & 18.455482 & 1.44071 & 0.030167 & \textbf{0.024532} \\
\hline \end{tabular}
            }
            \label{tab:random-lydia-case-06-50}
            \end{table*}
\begin{table*}[h]
                \centering
                \caption{Results on \emph{Random Syft 03, 100-200}, time in milliseconds.}
                \resizebox{0.95\linewidth}{!}{ \renewcommand{\arraystretch}{1}
\begin{tabular}{|c|c|c|c|c|c|c|c|c|c|c|c|c|}
\hline
\textbf{name} & \textbf{Nike Hash True First} & \textbf{Nike Hash False First} & \textbf{Nike Hash Random} & \textbf{Nike Bdd True First} & \textbf{Nike Bdd False First} & \textbf{Nike Bdd Random} & \textbf{Nike Multithreaded} & \textbf{Cynthia} & \textbf{Lydia} & \textbf{Lisa Symbolic} & \textbf{Lisa Explicit} & \textbf{Lisa} \\ \hline
101 & 371.12851 & 147.818981 & 10690.126604 & 397.076257 & \textbf{138.739578} & 2279.263585 & 152.011757 & 80864.237136 & 3575.094001 & 461.358 & --- & 1177.9 \\ \hline 
102 & 63.73732 & 97.00888 & 1652.763912 & 57.05748 & 85.436514 & 1128.156722 & 100.213576 & 3662.879195 & 9.482275 & --- & 0.020742 & \textbf{0.01983} \\ \hline 
103 & 3712.589305 & 3496.971662 & --- & 4047.902241 & 687.365672 & --- & 3557.506868 & --- & 5566.685148 & --- & 236.061 & \textbf{31.056} \\ \hline 
104 & 150.999029 & 197.437402 & 8991.704329 & 146.525042 & 122.209519 & 1841.166484 & 200.144001 & --- & 36.23162 & \textbf{0.343074} & 4.65217 & 4.464 \\ \hline 
105 & 8301.681749 & 1248.941079 & --- & 3000.670894 & 1129.567495 & 70261.395572 & 1264.778544 & --- & 13.784389 & --- & 0.021449 & \textbf{0.020968} \\ \hline 
106 & 36.962905 & 27.474348 & 606.17575 & 39.47189 & 24.850223 & 255.308462 & 29.385874 & 3547.757753 & 20.487993 & 0.442435 & 0.212219 & \textbf{0.193576} \\ \hline 
107 & 18.633927 & 17.111674 & 18.292596 & 18.541627 & 17.894219 & 18.878602 & 21.130126 & \textbf{1.122198} & 226.442828 & 103.841 & 33.6333 & 30.2401 \\
\hline \end{tabular}
            }
            \label{tab:random-syft-03-100-200}
            \end{table*}
\begin{table*}[h]
                \centering
                \caption{Results on \emph{Random Syft 03, 0-100}, time in milliseconds.}
                \resizebox{0.95\linewidth}{!}{ \renewcommand{\arraystretch}{1}
\begin{tabular}{|c|c|c|c|c|c|c|c|c|c|c|c|c|}
\hline
\textbf{name} & \textbf{Nike Hash True First} & \textbf{Nike Hash False First} & \textbf{Nike Hash Random} & \textbf{Nike Bdd True First} & \textbf{Nike Bdd False First} & \textbf{Nike Bdd Random} & \textbf{Nike Multithreaded} & \textbf{Cynthia} & \textbf{Lydia} & \textbf{Lisa Symbolic} & \textbf{Lisa Explicit} & \textbf{Lisa} \\ \hline
1 & 5188.103402 & 132263.301164 & --- & 4180.54409 & 2711.475735 & --- & 5342.515073 & --- & 78.270212 & \textbf{1.10837} & 6.9235 & 4.49524 \\ \hline 
2 & 105948.080231 & 1316.554382 & --- & 25574.821837 & 1300.172417 & --- & 1346.357127 & --- & 25.187672 & \textbf{0.220906} & 0.450101 & 0.396162 \\ \hline 
3 & 19.265656 & 19.662589 & 20.12297 & 19.936581 & 19.238229 & 19.713098 & 21.89371 & 1.448557 & 14.246965 & 0.063966 & \textbf{0.029014} & 0.033395 \\ \hline 
4 & 233.067626 & 243.093066 & 42490.74027 & 171.539432 & 186.165276 & 1894.872473 & 254.627278 & --- & 13.381828 & \textbf{0.159445} & 0.162832 & 0.163528 \\ \hline 
5 & 45.987631 & 43.572294 & 43.56121 & 49.200354 & 47.383417 & 47.982572 & 47.897064 & 351.33871 & 19.11144 & \textbf{3.26319} & 5.61274 & 5.38125 \\ \hline 
6 & 16225.575144 & 1317.145448 & --- & 3527.678918 & 1404.144196 & 48471.777664 & 1347.228569 & --- & 26.177785 & 0.066827 & 0.043084 & \textbf{0.039006} \\ \hline 
7 & 491.301194 & 227.206716 & 157146.942522 & 204.606569 & 88.506694 & 2392.068684 & 234.019052 & --- & 12.294267 & 0.061251 & 0.043303 & \textbf{0.031415} \\
\hline \end{tabular}
            }
            \label{tab:random-syft-03-0-100}
            \end{table*}
\begin{table*}[h]
                \centering
                \caption{Results on \emph{Random Lydia Case 04 50}, time in milliseconds.}
                \resizebox{0.95\linewidth}{!}{ \renewcommand{\arraystretch}{1}
\begin{tabular}{|c|c|c|c|c|c|c|c|c|c|c|c|c|}
\hline
\textbf{name} & \textbf{Nike Hash True First} & \textbf{Nike Hash False First} & \textbf{Nike Hash Random} & \textbf{Nike Bdd True First} & \textbf{Nike Bdd False First} & \textbf{Nike Bdd Random} & \textbf{Nike Multithreaded} & \textbf{Cynthia} & \textbf{Lydia} & \textbf{Lisa Symbolic} & \textbf{Lisa Explicit} & \textbf{Lisa} \\ \hline
1 & 16.870001 & 17.153821 & 17.18682 & 16.384968 & 17.227759 & 16.543762 & 19.17354 & 6886.058277 & 11.739812 & \textbf{0.328307} & 0.520087 & 0.516529 \\ \hline 
2 & 15.399163 & 15.196455 & 15.486768 & 16.118445 & 15.354612 & 15.305571 & 18.010372 & 0.384879 & 6.501435 & 0.058177 & 0.034233 & \textbf{0.032827} \\ \hline 
3 & 247.406245 & 18.096118 & 1644.213418 & 178.924989 & 18.780541 & 138.306082 & 20.774158 & 2303.136716 & 23.831855 & \textbf{0.089467} & 0.356991 & 0.337432 \\ \hline 
4 & 1773.638344 & 691.986975 & --- & 1593.325115 & 495.672686 & --- & 711.480355 & --- & 156043.561906 & \textbf{83.2831} & --- & 1869.97 \\ \hline 
5 & 3697.624058 & 578.107085 & 89662.888683 & 1848.740326 & 462.946393 & 31916.684263 & 587.466884 & --- & 1902.637411 & \textbf{6.17162} & 495.238 & 10.5653 \\ \hline 
6 & 19.20774 & 18.906754 & 19.334721 & 18.794765 & 19.034083 & 18.790621 & 20.62917 & 1.393927 & 10.059809 & 0.034501 & 0.027884 & \textbf{0.026008} \\ \hline 
7 & 1033.160623 & 705.34598 & --- & 1102.201835 & 500.411941 & --- & 720.373232 & --- & 196991.230935 & \textbf{79.4626} & --- & 1836.46 \\
\hline \end{tabular}
            }
            \label{tab:random-lydia-case-04-50}
            \end{table*}
\begin{table*}[h]
                \centering
                \caption{Results on \emph{Random Lydia Case 10 50}, time in milliseconds.}
                \resizebox{0.95\linewidth}{!}{ \renewcommand{\arraystretch}{1}
\begin{tabular}{|c|c|c|c|c|c|c|c|c|c|c|c|c|}
\hline
\textbf{name} & \textbf{Nike Hash True First} & \textbf{Nike Hash False First} & \textbf{Nike Hash Random} & \textbf{Nike Bdd True First} & \textbf{Nike Bdd False First} & \textbf{Nike Bdd Random} & \textbf{Nike Multithreaded} & \textbf{Cynthia} & \textbf{Lydia} & \textbf{Lisa Symbolic} & \textbf{Lisa Explicit} & \textbf{Lisa} \\ \hline
1 & --- & --- & --- & --- & --- & --- & --- & --- & 28.782413 & 60.3576 & \textbf{0.032948} & 0.037117 \\ \hline 
2 & --- & --- & --- & --- & --- & --- & --- & --- & --- & --- & --- & --- \\ \hline 
3 & 17.575102 & 17.352167 & 18.173289 & 17.371955 & 17.395569 & 17.545897 & 18.723553 & 1.080734 & 140.346973 & \textbf{0.200886} & 0.597075 & 0.572767 \\ \hline 
4 & --- & \textbf{574.387805} & --- & --- & 981.787471 & --- & 587.057097 & --- & 157430.510276 & 860.009 & --- & 1093.0 \\ \hline 
5 & --- & --- & --- & --- & --- & --- & --- & --- & --- & --- & --- & --- \\ \hline 
6 & 17.371103 & 17.750082 & 17.536003 & 17.546889 & 17.228279 & 17.089149 & 18.940858 & 1.191698 & 141.895564 & \textbf{0.197331} & 0.64155 & 0.526382 \\ \hline 
7 & --- & \textbf{297.350363} & --- & --- & 565.279784 & --- & 303.154461 & --- & 83451.112641 & 1111.48 & --- & 5069.74 \\
\hline \end{tabular}
            }
            \label{tab:random-lydia-case-10-50}
            \end{table*}
\begin{table*}[h]
                \centering
                \caption{Results on \emph{Nim 01}, time in milliseconds.}
                \resizebox{0.95\linewidth}{!}{ \renewcommand{\arraystretch}{1}
\begin{tabular}{|c|c|c|c|c|c|c|c|c|c|c|c|c|}
\hline
\textbf{name} & \textbf{Nike Hash True First} & \textbf{Nike Hash False First} & \textbf{Nike Hash Random} & \textbf{Nike Bdd True First} & \textbf{Nike Bdd False First} & \textbf{Nike Bdd Random} & \textbf{Nike Multithreaded} & \textbf{Cynthia} & \textbf{Lydia} & \textbf{Lisa Symbolic} & \textbf{Lisa Explicit} & \textbf{Lisa} \\ \hline
nim\_01\_01 & 21.652504 & 21.852993 & 22.425201 & 22.51097 & 22.575411 & 22.619657 & 22.869925 & 10.701386 & 9.243031 & 0.511151 & 0.054589 & \textbf{0.039505} \\ \hline 
nim\_01\_02 & 36.510922 & 32.635422 & 29.130419 & 51.644759 & 46.93518 & 42.953328 & 33.937856 & 34.209743 & 14.044925 & 0.783896 & \textbf{0.107531} & 0.108793 \\ \hline 
nim\_01\_03 & 124.350392 & 748.812837 & 71.802001 & 147.81038 & 500.318885 & 94.538401 & 124.496853 & 50.863318 & 28.067614 & 1.39775 & 0.165809 & \textbf{0.159741} \\ \hline 
nim\_01\_04 & 655.990604 & 5078.174691 & 346.277822 & 478.273479 & 3572.571717 & 175.8517 & 636.957444 & 51.149021 & 148.985605 & 1.90265 & 0.253887 & \textbf{0.216297} \\ \hline 
nim\_01\_05 & 3951.684383 & 31716.609289 & 2088.068016 & 1893.314961 & 23764.470102 & 322.293758 & 3939.678555 & 83.136903 & 648.186875 & 2.67965 & 0.354514 & \textbf{0.33663} \\ \hline 
nim\_01\_06 & 25804.192511 & 185198.807923 & 14973.912576 & 8558.792909 & 146791.191633 & 607.119617 & 27804.395652 & 117.085951 & 2835.283053 & 4.26119 & \textbf{0.491802} & 0.509569 \\ \hline 
nim\_01\_07 & 180527.335506 & --- & 104195.615515 & 40304.12107 & --- & 1784.890989 & 106272.648966 & 173.173464 & 12656.386116 & 6.0989 & \textbf{0.713635} & 0.713666 \\
\hline \end{tabular}
            }
            \label{tab:nim-01}
            \end{table*}
\begin{table*}[h]
                \centering
                \caption{Results on \emph{Random Lydia Case 08 50}, time in milliseconds.}
                \resizebox{0.95\linewidth}{!}{ \renewcommand{\arraystretch}{1}
\begin{tabular}{|c|c|c|c|c|c|c|c|c|c|c|c|c|}
\hline
\textbf{name} & \textbf{Nike Hash True First} & \textbf{Nike Hash False First} & \textbf{Nike Hash Random} & \textbf{Nike Bdd True First} & \textbf{Nike Bdd False First} & \textbf{Nike Bdd Random} & \textbf{Nike Multithreaded} & \textbf{Cynthia} & \textbf{Lydia} & \textbf{Lisa Symbolic} & \textbf{Lisa Explicit} & \textbf{Lisa} \\ \hline
1 & --- & --- & --- & --- & 202624.173201 & --- & --- & --- & 25.582577 & 32.5515 & 0.067548 & \textbf{0.033464} \\ \hline 
2 & --- & --- & --- & --- & --- & --- & --- & --- & --- & --- & --- & --- \\ \hline 
3 & 18.689927 & 18.482602 & 19.206651 & 18.446774 & \textbf{18.401904} & 18.64912 & 19.907831 & --- & 1004.314178 & 22.5193 & 35900.2 & 217.157 \\ \hline 
4 & 18.556394 & 19.312037 & 19.045033 & 17.92455 & 19.07326 & 19.28501 & 21.067634 & \textbf{1.790123} & --- & --- & --- & --- \\ \hline 
5 & --- & --- & --- & --- & --- & --- & --- & --- & --- & --- & --- & --- \\ \hline 
6 & --- & 102664.739203 & --- & --- & 71304.804389 & --- & --- & --- & --- & \textbf{16434.9} & --- & --- \\ \hline 
7 & 35914.243109 & 48.799137 & --- & 22568.16775 & 32.73514 & --- & 51.666327 & --- & 2649.604001 & \textbf{4.06637} & --- & 15.3771 \\
\hline \end{tabular}
            }
            \label{tab:random-lydia-case-08-50}
            \end{table*}
\begin{table*}[h]
                \centering
                \caption{Results on \emph{Random Lydia Case 09 50}, time in milliseconds.}
                \resizebox{0.95\linewidth}{!}{ \renewcommand{\arraystretch}{1}
\begin{tabular}{|c|c|c|c|c|c|c|c|c|c|c|c|c|}
\hline
\textbf{name} & \textbf{Nike Hash True First} & \textbf{Nike Hash False First} & \textbf{Nike Hash Random} & \textbf{Nike Bdd True First} & \textbf{Nike Bdd False First} & \textbf{Nike Bdd Random} & \textbf{Nike Multithreaded} & \textbf{Cynthia} & \textbf{Lydia} & \textbf{Lisa Symbolic} & \textbf{Lisa Explicit} & \textbf{Lisa} \\ \hline
1 & 200.296323 & 201.567275 & 203.168077 & 212.315198 & 211.333889 & 215.085183 & 214.71067 & --- & 3765.572701 & \textbf{97.1977} & --- & 880.774 \\ \hline 
2 & --- & --- & --- & --- & --- & --- & --- & --- & --- & --- & --- & --- \\ \hline 
3 & 20.147491 & 19.405688 & 20.195466 & 19.662369 & 20.110029 & 20.140943 & 21.846456 & \textbf{12.447289} & --- & --- & --- & --- \\ \hline 
4 & 996.878624 & \textbf{44.910837} & 13744.3846 & 1099.670634 & 46.319641 & 12883.353099 & 48.160912 & --- & 933.194199 & 95.7317 & 713.011 & 834.38 \\ \hline 
5 & --- & --- & --- & --- & --- & --- & --- & --- & --- & --- & --- & --- \\ \hline 
6 & 19.585303 & 20.233692 & 19.673646 & 19.668884 & 19.37875 & 19.736536 & 21.890737 & \textbf{11.056458} & --- & --- & --- & --- \\ \hline 
7 & --- & --- & --- & --- & --- & --- & --- & --- & --- & --- & --- & --- \\
\hline \end{tabular}
            }
            \label{tab:random-lydia-case-09-50}
            \end{table*}
\begin{table*}[h]
                \centering
                \caption{Results on \emph{Double Counter}, time in milliseconds.}
                \resizebox{0.95\linewidth}{!}{ \renewcommand{\arraystretch}{1}
\begin{tabular}{|c|c|c|c|c|c|c|c|c|c|c|c|c|}
\hline
\textbf{name} & \textbf{Nike Hash True First} & \textbf{Nike Hash False First} & \textbf{Nike Hash Random} & \textbf{Nike Bdd True First} & \textbf{Nike Bdd False First} & \textbf{Nike Bdd Random} & \textbf{Nike Multithreaded} & \textbf{Cynthia} & \textbf{Lydia} & \textbf{Lisa Symbolic} & \textbf{Lisa Explicit} & \textbf{Lisa} \\ \hline
counters\_01 & 18.596958 & 18.268062 & 18.403466 & 19.253715 & 19.107806 & 19.412668 & 19.18572 & 34.500707 & 9.291831 & 10.6135 & \textbf{0.190149} & 0.225103 \\ \hline 
counters\_02 & 20.013578 & 20.253191 & 19.492336 & 94.904176 & 76.403862 & 76.833945 & 21.537058 & 90.045937 & 26.0702 & 36.6841 & 0.720083 & \textbf{0.717793} \\ \hline 
counters\_03 & 26.794465 & 25.526104 & 25.295872 & 209.794823 & 208.677393 & 207.428727 & 26.971165 & --- & 184.525338 & 365.612 & \textbf{3.43736} & 41.6974 \\ \hline 
counters\_04 & 38.482922 & 37.337264 & \textbf{37.016285} & 395.887396 & 423.36698 & 486.792135 & 39.26482 & --- & 2540.111074 & 3641.88 & 45.9793 & 72.3087 \\ \hline 
counters\_05 & \textbf{64.703472} & 65.006299 & 64.764067 & 836.840309 & 881.468907 & 1140.178069 & 66.213866 & --- & 42335.932599 & --- & 274.254 & 535.203 \\ \hline 
counters\_06 & \textbf{124.922986} & 126.558988 & 127.267784 & 2382.959143 & 1774.30815 & 1982.499516 & 126.82075 & --- & --- & --- & --- & --- \\ \hline 
counters\_07 & 286.805523 & \textbf{286.526689} & 289.076264 & 3131.282049 & 3796.194518 & 4199.210945 & 296.109436 & --- & --- & --- & --- & --- \\
\hline \end{tabular}
            }
            \label{tab:double-counter}
            \end{table*}
\begin{table*}[h]
                \centering
                \caption{Results on \emph{Nim 02}, time in milliseconds.}
                \resizebox{0.95\linewidth}{!}{ \renewcommand{\arraystretch}{1}
\begin{tabular}{|c|c|c|c|c|c|c|c|c|c|c|c|c|}
\hline
\textbf{name} & \textbf{Nike Hash True First} & \textbf{Nike Hash False First} & \textbf{Nike Hash Random} & \textbf{Nike Bdd True First} & \textbf{Nike Bdd False First} & \textbf{Nike Bdd Random} & \textbf{Nike Multithreaded} & \textbf{Cynthia} & \textbf{Lydia} & \textbf{Lisa Symbolic} & \textbf{Lisa Explicit} & \textbf{Lisa} \\ \hline
nim\_02\_01 & 86.975632 & 112.611648 & 105.373406 & 98.569137 & 119.114981 & 119.995565 & 114.925515 & 115.180326 & 22.998994 & 7.71421 & 0.335915 & \textbf{0.32857} \\ \hline 
nim\_02\_02 & 14861.324807 & 16186.670998 & 14325.332997 & 7333.083561 & 8301.186367 & 7004.188455 & 16339.154739 & 730.429379 & 291.899614 & 42.5965 & 1.10907 & \textbf{1.10246} \\ \hline 
nim\_02\_03 & 258217.922498 & --- & --- & 107711.86067 & 200600.7252 & 168324.776357 & --- & 1018.461399 & 4443.044252 & 178.374 & \textbf{2.8325} & 2.90445 \\ \hline 
nim\_02\_04 & --- & --- & --- & --- & --- & --- & --- & 5464.305582 & 73887.195858 & 287.6 & \textbf{6.57803} & 7.06112 \\ \hline 
nim\_02\_05 & --- & --- & --- & --- & --- & --- & --- & \textbf{22790.095015} & --- & --- & --- & --- \\ \hline 
nim\_02\_06 & --- & --- & --- & --- & --- & --- & --- & \textbf{37507.865014} & --- & --- & --- & --- \\ \hline 
nim\_02\_07 & --- & --- & --- & --- & --- & --- & --- & \textbf{23257.818323} & --- & --- & --- & --- \\ \hline 
nim\_02\_08 & --- & --- & --- & --- & --- & --- & --- & \textbf{169177.832518} & --- & --- & --- & --- \\ \hline 
nim\_02\_09 & --- & --- & --- & --- & --- & --- & --- & \textbf{95835.707556} & --- & --- & --- & --- \\ \hline 
nim\_02\_10 & --- & --- & --- & --- & --- & --- & --- & --- & --- & --- & --- & --- \\ \hline \end{tabular}
            }
            \label{tab:nim-02}
            \end{table*}
\begin{table*}[h]
                \centering
                \caption{Results on \emph{Nim 05}, time in milliseconds.}
                \resizebox{0.95\linewidth}{!}{ \renewcommand{\arraystretch}{1}
\begin{tabular}{|c|c|c|c|c|c|c|c|c|c|c|c|c|}
\hline
\textbf{name} & \textbf{Nike Hash True First} & \textbf{Nike Hash False First} & \textbf{Nike Hash Random} & \textbf{Nike Bdd True First} & \textbf{Nike Bdd False First} & \textbf{Nike Bdd Random} & \textbf{Nike Multithreaded} & \textbf{Cynthia} & \textbf{Lydia} & \textbf{Lisa Symbolic} & \textbf{Lisa Explicit} & \textbf{Lisa} \\ \hline
nim\_05\_01 & --- & --- & --- & --- & --- & --- & --- & \textbf{220651.046985} & --- & --- & --- & --- \\ \hline 
nim\_05\_02 & --- & --- & --- & --- & --- & --- & --- & --- & --- & --- & --- & --- \\ \hline 
nim\_05\_03 & --- & --- & --- & --- & --- & --- & --- & --- & --- & --- & --- & --- \\ \hline 
nim\_05\_04 & --- & --- & --- & --- & --- & --- & --- & --- & --- & --- & --- & --- \\ \hline 
nim\_05\_05 & --- & --- & --- & --- & --- & --- & --- & --- & --- & --- & --- & --- \\ \hline 
nim\_05\_06 & --- & --- & --- & --- & --- & --- & --- & --- & --- & --- & --- & --- \\ \hline 
nim\_05\_07 & --- & --- & --- & --- & --- & --- & --- & --- & --- & --- & --- & --- \\ \hline 
nim\_05\_08 & --- & --- & --- & --- & --- & --- & --- & --- & --- & --- & --- & --- \\ \hline 
nim\_05\_09 & --- & --- & --- & --- & --- & --- & --- & --- & --- & --- & --- & --- \\ \hline \end{tabular}
            }
            \label{tab:nim-05}
            \end{table*}
\begin{table*}[h]
                \centering
                \caption{Results on \emph{Nim 03}, time in milliseconds.}
                \resizebox{0.95\linewidth}{!}{ \renewcommand{\arraystretch}{1}
\begin{tabular}{|c|c|c|c|c|c|c|c|c|c|c|c|c|}
\hline
\textbf{name} & \textbf{Nike Hash True First} & \textbf{Nike Hash False First} & \textbf{Nike Hash Random} & \textbf{Nike Bdd True First} & \textbf{Nike Bdd False First} & \textbf{Nike Bdd Random} & \textbf{Nike Multithreaded} & \textbf{Cynthia} & \textbf{Lydia} & \textbf{Lisa Symbolic} & \textbf{Lisa Explicit} & \textbf{Lisa} \\ \hline
nim\_03\_01 & 14589.428787 & 15583.749344 & 17541.842061 & 7563.240427 & 7376.988988 & 8799.891142 & 15761.988348 & 650.571523 & 270.715502 & 473.228 & \textbf{0.26237} & 0.27006 \\ \hline 
nim\_03\_02 & --- & --- & --- & --- & --- & --- & --- & 12655.221579 & 42245.551844 & 3102.4 & 24.4585 & \textbf{8.01435} \\ \hline 
nim\_03\_03 & --- & --- & --- & --- & --- & --- & --- & --- & --- & --- & --- & --- \\ \hline \end{tabular}
            }
            \label{tab:nim-03}
            \end{table*}
\begin{table*}[h]
                \centering
                \caption{Results on \emph{Single Counter}, time in milliseconds.}
                \resizebox{0.95\linewidth}{!}{ \renewcommand{\arraystretch}{1}
\begin{tabular}{|c|c|c|c|c|c|c|c|c|c|c|c|c|}
\hline
\textbf{name} & \textbf{Nike Hash True First} & \textbf{Nike Hash False First} & \textbf{Nike Hash Random} & \textbf{Nike Bdd True First} & \textbf{Nike Bdd False First} & \textbf{Nike Bdd Random} & \textbf{Nike Multithreaded} & \textbf{Cynthia} & \textbf{Lydia} & \textbf{Lisa Symbolic} & \textbf{Lisa Explicit} & \textbf{Lisa} \\ \hline
counter\_01 & 18.058177 & 18.350703 & 18.060843 & 18.169518 & 18.933748 & 18.692677 & 19.607023 & 16.453027 & 6.978378 & 0.2575 & 0.118257 & \textbf{0.112888} \\ \hline 
counter\_02 & 24.204002 & 31.358584 & 23.964412 & 47.131896 & 50.575059 & 47.358771 & 32.82203 & 211.518911 & 8.89128 & 2.61965 & \textbf{0.280484} & 0.288077 \\ \hline 
counter\_03 & 67.810516 & 114.23335 & 104.986782 & 134.0321 & 171.003919 & 150.686689 & 119.903069 & 693.556915 & 12.781949 & 90.6461 & 16.9257 & \textbf{0.889436} \\ \hline 
counter\_04 & 428.059632 & 831.486601 & 698.842123 & 588.71692 & 944.899279 & 944.688069 & 787.3057 & 2522.781763 & 25.034847 & 1218.46 & 27.5305 & \textbf{3.34075} \\ \hline 
counter\_05 & 3754.291805 & 7568.91302 & 6852.560534 & 4065.36963 & 7778.351135 & 5770.683935 & 6975.687144 & --- & 67.723467 & 13715.5 & \textbf{58.4356} & 381.417 \\ \hline 
counter\_06 & 34590.517551 & 70189.903347 & 51550.536626 & 35831.825415 & 70901.656923 & 48650.17893 & 65032.73378 & --- & \textbf{248.148853} & 143515.0 & 484.092 & 3020.43 \\ \hline 
counter\_07 & --- & --- & --- & --- & --- & --- & --- & --- & \textbf{1107.050207} & --- & 4089.62 & 27384.9 \\
\hline \end{tabular}
            }
            \label{tab:single-counter}
            \end{table*}

\end{document}